\documentclass[12pt]{iopart}

\usepackage{mathrsfs}

\usepackage{graphicx}
\usepackage{dcolumn}
\usepackage{bm}

\expandafter\let\csname equation*\endcsname\undefined
\expandafter\let\csname endequation*\endcsname\undefined
\usepackage{amsmath}

\usepackage{amsthm}
\usepackage{amssymb}
\usepackage{xfrac}
\usepackage[mathcal]{euscript}
\usepackage{tikz}
\usetikzlibrary{arrows.meta,shapes.arrows}
\usepackage{mathtools}

\usepackage{url}
\usepackage{hyperref}
\usepackage{color}
\definecolor{refcolor}{RGB}{0,0,190}
\hypersetup{
    colorlinks,
    citecolor=refcolor,
    filecolor=refcolor,
    linkcolor=refcolor,
    urlcolor=refcolor
}

\usepackage{booktabs}

\usepackage{bookmark}
\bookmarksetup{
  numbered, 
  open,
}

\theoremstyle{definition}
\newtheorem{theorem}{Theorem}
\newtheorem{proposition}{Proposition}

\newtheorem{corollary}{Corollary}

\theoremstyle{definition}
\newtheorem{remark}{Remark}
\newtheorem{example}{Example}

\newtheorem{type}{Type}

\theoremstyle{definition}
\newtheorem{observation}{Observation}

\newtheorem{definition}{Definition}

\makeatletter
\newcommand{\neutralize}[1]{\expandafter\let\csname c@#1\endcsname\count@}
\makeatother

\newtheorem{objection}{Objection}
\newenvironment{objectionbis}[1]
  {\renewcommand{\theobjection}{\ref*{#1}$'$}%
	 \neutralize{objection}\phantomsection
   \begin{objection}}
  {\end{objection}}
	
\newtheorem{reply}{Reply to Objection}
\newenvironment{replyToObjection}[1]
  {\renewcommand{\thereply}{ \ref*{#1}}%
	 \neutralize{reply}\phantomsection
   \begin{reply}}
  {\end{reply}}

\theoremstyle{definition}
\newtheorem{defCustom}{Definition}
\renewcommand{\thedefCustom}{\arabic{definition}}
\makeatletter
\newcommand{\setdefCustomtag}[1]{
  \let\oldthedefCustom\thedefCustom
  \renewcommand{\thedefCustom}{#1}
  \g@addto@macro\enddefCustom{
    \global\let\thedefCustom\oldthedefCustom}
  }
\makeatother

\newtheorem{thesisCustom}{Thesis}
\renewcommand{\thethesisCustom}{\arabic{thesis}}
\makeatletter
\newcommand{\setthesisCustomtag}[1]{
  \let\oldthethesisCustom\thethesisCustom
  \renewcommand{\thethesisCustom}{#1}
  \g@addto@macro\endthesisCustom{
    \global\let\thethesisCustom\oldthethesisCustom}
  }
\makeatother

\newtheorem{condition}{Condition}
\renewcommand{\thecondition}{\arabic{condition}}
\makeatletter
\newcommand{\setconditiontag}[1]{
  \let\oldthecondition\thecondition
  \renewcommand{\thecondition}{#1}
  \g@addto@macro\endcondition{
    \global\let\thecondition\oldthecondition}
  }
\makeatother

\newtheorem{assumption}{Assumption}
\renewcommand{\theassumption}{\arabic{assumption}}
\makeatletter
\newcommand{\setassumptiontag}[1]{
  \let\oldtheassumption\theassumption
  \renewcommand{\theassumption}{#1}
  \g@addto@macro\endassumption{
    \global\let\theassumption\oldtheassumption}
  }
\makeatother

\newtheorem{claim}{Claim}
\renewcommand{\theclaim}{\arabic{claim}}
\makeatletter
\newcommand{\setclaimtag}[1]{
  \let\oldtheclaim\theclaim
  \renewcommand{\theclaim}{#1}
  \g@addto@macro\endclaim{
    \global\let\theclaim\oldtheclaim}
  }
\makeatother

\theoremstyle{remark}
\newtheorem{pointItem}{}
\renewcommand{\thepointItem}{\quad\arabic{pointItem}}
\makeatletter
\newcommand{\setpointItemtag}[1]{
  \let\oldthepointItem\thepointItem
  \renewcommand{\thepointItem}{#1}
  \g@addto@macro\endpointItem{
    \global\let\thepointItem\oldthepointItem}
  }
\makeatother

\begin{document}


\newcommand{\pbref}[1]{\ref{#1} (\nameref*{#1})}

\def\({\big(}
\def\){\big)}

\newcommand{\tn}{\textnormal}
\newcommand{\ds}{\displaystyle}
\newcommand{\dsfrac}[2]{\displaystyle{\frac{#1}{#2}}}

\newcommand{\boplus}{\textstyle{\bigoplus}}
\newcommand{\botimes}{\textstyle{\bigotimes}}
\newcommand{\bcup}{\textstyle{\bigcup}}
\newcommand{\bsqcup}{\textstyle{\bigsqcup}}
\newcommand{\bcap}{\textstyle{\bigcap}}

\newcommand{\struct}{\mc{S}}
\newcommand{\kind}{\mc{K}}

\newcommand{\dddots}{\rotatebox[origin=t]{135}{$\cdots$}}

\newcommand{\statespace}{\mathcal{S}}
\newcommand{\hilbert}{\mathcal{H}}
\newcommand{\vectorspace}{\mathcal{V}}
\newcommand{\mc}[1]{\mathcal{#1}}
\newcommand{\mf}[1]{\mathfrak{#1}}
\newcommand{\dU}{\wh{\mc{U}}}

\newcommand{\wh}[1]{\widehat{#1}}
\newcommand{\dwh}[1]{\wh{\rule{0ex}{1.3ex}\smash{\wh{\hfill{#1}\,}}}}

\newcommand{\wt}[1]{\widetilde{#1}}
\newcommand{\wht}[1]{\widehat{\widetilde{#1}}}

\newcommand{\qmU}{$\mathscr{U}$}
\newcommand{\qmR}{$\mathscr{R}$}
\newcommand{\qmUR}{$\mathscr{UR}$}
\newcommand{\qmDR}{$\mathscr{DR}$}

\newcommand{\R}{\mathbb{R}}
\newcommand{\C}{\mathbb{C}}
\newcommand{\Z}{\mathbb{Z}}
\newcommand{\K}{\mathbb{K}}
\newcommand{\N}{\mathbb{N}}
\newcommand{\Prj}{\mathcal{P}}
\newcommand{\abs}[1]{|#1|}

\newcommand{\de}{\text{d}}

\newcommand{\ie}{\textit{i.e.}\ }
\newcommand{\vs}{\textit{vs.}\ }
\newcommand{\eg}{\textit{e.g.}\ }
\newcommand{\cf}{\textit{cf.}\ }
\newcommand{\etc}{\textit{etc}}

\newcommand{\Span}{\tn{span}}
\newcommand{\pde}{PDE}
\newcommand{\U}{\tn{U}}
\newcommand{\SU}{\tn{SU}}
\newcommand{\GL}{\tn{GL}}

\newcommand{\schrod}{Schr\"odinger}
\newcommand{\vonneum}{Liouville-von Neumann}
\newcommand{\ks}{Kochen-Specker}
\newcommand{\leggarg}{Leggett-Garg}
\newcommand{\bra}[1]{\langle#1|}
\newcommand{\ket}[1]{|#1\rangle}
\newcommand{\kett}[1]{|\!\!|#1\rangle\!\!\rangle}
\newcommand{\proj}[1]{\ket{#1}\bra{#1}}
\newcommand{\braket}[2]{\langle#1|#2\rangle}
\newcommand{\ketbra}[2]{|#1\rangle\langle#2|}
\newcommand{\expectation}[1]{\langle#1\rangle}
\newcommand{\Herm}{\tn{Herm}}
\newcommand{\Sym}[1]{\tn{Sym}_{#1}}
\newcommand{\meanvalue}[2]{\langle{#1}\rangle_{#2}}
\newcommand{\Prob}{\tn{Prob}}
\newcommand{\kjj}[3]{#1\!:\!#2,#3}
\newcommand{\jk}[2]{#1,#2}
\newcommand{\JK}{\mf{j}}

\newcommand{\weightU}[5]{\big[{#2}{}_{#3}\overset{#1}{\rightarrow}{#4}{}_{#5}\big]}
\newcommand{\weightUT}[8]{\big[{#3}{}_{#4}\overset{#1}{\rightarrow}{#5}{}_{#6}\overset{#2}{\rightarrow}{#7}{}_{#8}\big]}
\newcommand{\weight}[4]{\weightU{}{#1}{#2}{#3}{#4}}
\newcommand{\weightT}[6]{\weightUT{}{}{#1}{#2}{#3}{#4}{#5}{#6}}

\newcommand{\btimes}{\boxtimes}
\newcommand{\btimess}{{\boxtimes_s}}

\newcommand{\h}{\mathbf{(2\pi\hbar)}}
\newcommand{\x}{\mathbf{x}}
\newcommand{\A}{\mathbf{a}}
\newcommand{\xThree}{\boldsymbol{x}}
\newcommand{\z}{\mathbf{z}}
\newcommand{\q}{\mathbf{q}}
\newcommand{\p}{\mathbf{p}}
\newcommand{\0}{\mathbf{0}}
\newcommand{\annih}{\widehat{\mathbf{a}}}

\newcommand{\cs}{\mathscr{C}}
\newcommand{\ps}{\mathscr{P}}
\newcommand{\xhat}{\widehat{\x}}
\newcommand{\phat}{\widehat{\mathbf{p}}}
\newcommand{\fqproj}[1]{\Pi_{#1}}
\newcommand{\cqproj}[1]{\wh{\Pi}_{#1}}
\newcommand{\cproj}[1]{\wh{\Pi}^{\perp}_{#1}}

\newcommand{\M}{\mathbb{E}_3}
\newcommand{\D}{\mathbf{D}}
\newcommand{\dn}{\tn{d}}
\newcommand{\db}{\mathbf{d}}
\newcommand{\n}{\mathbf{n}}
\newcommand{\m}{\mathbf{m}}
\newcommand{\V}[1]{\mathbb{V}_{#1}}
\newcommand{\F}[1]{\mathcal{F}_{#1}}
\newcommand{\Fvacuumfield}{\widetilde{\mathcal{F}}^0}
\newcommand{\nD}[1]{|{#1}|}
\newcommand{\Lin}{\mathcal{L}}
\newcommand{\End}{\tn{End}}
\newcommand{\vbundle}[4]{{#1}\to {#2} \stackrel{\pi_{#3}}{\to} {#4}}
\newcommand{\vbundlex}[1]{\vbundle{V_{#1}}{E_{#1}}{#1}{M_{#1}}}
\newcommand{\rep}{\rho_{\scriptscriptstyle\btimes}}

\newcommand{\intl}[1]{\int\limits_{#1}}

\newcommand{\moyalBracket}[1]{\{\mskip-5mu\{#1\}\mskip-5mu\}}

\newcommand{\Hint}{H_{\tn{int}}}

\newcommand{\quot}[1]{``#1''}

\def\sref #1{\S\ref{#1}}

\newcommand{\dBB}{de Broglie--Bohm}
\newcommand{\dBBt}{{\dBB} theory}
\newcommand{\pwt}{pilot-wave theory}
\newcommand{\PWT}{PWT}
\newcommand{\NRQM}{{\textbf{NRQM}}}

\newcommand{\TQS}{\ref{def:TQS}}
\newcommand{\HSF}{\ref{thesis:HSF}}

\newcommand{\todo}[1]{\textcolor{red}{$\Rightarrow$} \textcolor{blue}{#1}\PackageWarning{TODO:}{#1!}}

\clearpage

\title[Empirical adequacy of the time operator]{Empirical adequacy of the time operator canonically conjugate to a Hamiltonian generating translations}

\author{Ovidiu Cristinel Stoica}
\address{
 Dept. of Theoretical Physics, NIPNE---HH, Bucharest, Romania. \\
	Email: \href{mailto:cristi.stoica@theory.nipne.ro}{cristi.stoica@theory.nipne.ro},  \href{mailto:holotronix@gmail.com}{holotronix@gmail.com}
	}%

\date{\today}

\begin{abstract}
To admit a canonically conjugate time operator, the Hamiltonian has to be a generator of translations (like the momentum operator generates translations in space), so its spectrum must be unbounded. But the Hamiltonian governing our world is thought to be bounded from below. Also, judging by the number of fields and parameters of the Standard Model, the Hamiltonian seems much more complicated.

In this article I give examples of worlds governed by Hamiltonians generating translations. They can be expressed as a partial derivative operator just like the momentum operator, but when expressed in function of other observables they can exhibit any level of complexity. The examples include any quantum world realizing a standard ideal measurement, any quantum world containing a clock or a free massless fermion, the quantum representation of any deterministic time-reversible dynamical system without time loops, and any quantum world that cannot return to a past state.

Such worlds are as sophisticated as our world, but they admit a time operator. I show that, despite having unbounded Hamiltonian, they do not decay to infinite negative energy any more than any quantum or classical world.

Since two such quantum systems of the same Hilbert space dimension are unitarily equivalent even if the physical content of their observables is very different, they are concrete counterexamples to Hilbert Space Fundamentalism (HSF). Taking the observables into account removes the ambiguity of HSF and the clock ambiguity problem attributed to the Page-Wootters formalism, also caused by assuming HSF.

These results provide additional motivations to restore the spacetime symmetry in the formulation of Quantum Mechanics and for the Page-Wootters formalism.
\end{abstract}

\vspace{2pc}

\maketitle

\section{Introduction}
\label{s:intro}

Unlike the Theory of Relativity, Quantum Mechanics seems to treat space and time completely differently. While the position operator $\wh{x}$ is the canonical conjugate of the momentum operator $\wh{p}_x$, a generic Hamiltonian operator $\wh{H}$ does not admit a self-adjoint canonically conjugate operator $\wh{\tau}$ to represent time as a dynamical variable or observable. According to Pauli (\cite{Pauli1980GeneralPrinciplesOfQuantumMechanics}, footnote 2, page 63), 
\begin{quote}
In the older literature on Quantum Mechanics, we often find the operator equation
\begin{equation*}
H t-t H=\frac{h}{i}I,
\end{equation*}
[...] It is generally not possible, however, to construct a Hermitian operator (e.g. as a function of $p$ and $q$) which satisfies this equation. This is so because, from the C.R. [commutation relation] written above, it follows that $H$ possesses continuously all eigenvalues from $-\infty$ to $+\infty$ (cf. Dirac, Quantum Mechanics, First edition (1930), 34 and 56) whereas on the other hand, discrete eigenvalues of $H$ can be present.
\emph{We, therefore, conclude that the introduction of an operator $t$ is basically forbidden and the time $t$ must necessarily be considered as an ordinary number [...] in Quantum Mechanics (cf. for this E. {\schrod}, Berl. Ber. (1931) p. 238).}
\end{quote}

Since the discovery of the position-momentum energy relations by Heisenberg \cite{Heisenberg1927Uncertainty}, it was assumed that similar relations should hold between time and energy, due to the similar role played by space and time in relativity.
This belief was criticized by Pauli \cite{Pauli1980GeneralPrinciplesOfQuantumMechanics}.
Nevertheless, forms of time-energy uncertainty that do not violate Pauli's observation were shown to be possible, due to the fact that they can be interpreted to refer to various context, for example involving the characteristic times of average expectation values
\cite{MandelstamTamm1945TheUncertaintyRelationBetweenEnergyAndTimeInNonrelativisticQuantumMechanics}, the \emph{time of arrival} \cite{Razavy1971TimeOfArrivalOperator}, \cite{MugaEtal2007TimeInQuantumMechanicsI}, and the \emph{traversal} or \emph{tunneling time} (found to require, apparently, exceeding the speed of light) \cite{Hartman1962TunnelingOfAWavePacket}. Wigner clarified that the time-energy uncertainty relations should be understood to apply to the life-time of a definite state of a system, and examined the contrast with Heisenberg's position-momentum uncertainty \cite{Wigner1988OnTheTimeEnergyUncertaintyRelation}.
Aharonov and Bohm found situations where
the energy can be measured with arbitrary precision in an arbitrarily short time interval \cite{AharonovBohm1961TimeInTheQuantumTheoryAndTheUncertaintyRelationForTimeAndEnergy},
and Fock contradicted this on the grounds that it assumes ``the introduction of a field which does not obey the uncertainty relation'' \cite{Fock1962CriticismOfAnAttemptToDisproveTheUncertaintyRelationBetweenTimeAndEnergy}.
Time-energy uncertainty relations based on positive operator-valued measures are discussed in \cite{BuschGrabowskiLahti1995OperationalQuantumPhysics,Busch2002Time-energy-uncertainty-relation}.
For clarifications of various meanings of the time-energy uncertainty relations see Hilgevoord \cite{Hilgevoord1996TheUncertaintyPrincipleForEnergyAndTimeI,Hilgevoord1998TheUncertaintyPrincipleForEnergyAndTimeII,Hilgevoord2005TimeInQuantumMechanicsAStoryOfConfusion}.
There are several thorough reviews of these problems \cite{Jammer1974ThePhilosophyOfQuantumMechanics,BauerMello1978TheTimeEnergyUncertaintyRelation,Busch1990OnTheEnergyTimeUncertaintyRelationI,Busch1990OnTheEnergyTimeUncertaintyRelationII,Hilgevoord2005TimeInQuantumMechanicsAStoryOfConfusion,AltaieBeigeHodgson2022TimeAndQuantumClocksAReviewOfRecentDevelopments}, also see some important collections like \cite{NozKim1988SpecialRelativityAndQuantumTheoryACollectionOfPapersOnThePoincareGroup,MugaEtal2007TimeInQuantumMechanicsI,MugaEtal2009TimeInQuantumMechanicsII}.

Here we will explore the properties of quantum systems whose  Hamiltonian admits a canonical conjugate observable that can be understood as time, territory considered ``basically forbidden'' by Pauli, as per the above quote. The aim is to see whether this Hamiltonian can govern the dynamics of a world like the one we observe.
\setdefCustomtag{TQS}
\begin{defCustom}
\label{def:TQS}
A \emph{translational quantum system} is a closed quantum system with the Hamiltonian operator of the form
\begin{equation}
\label{eq:time_translation_Hamiltonian}
\wh{H}_{\tau}=-i\hbar\frac{\partial\ }{\partial\tau}.
\end{equation}

The evolution of the state vectors $\ket{\psi(t)}$ in the Hilbert space $\hilbert$ is given, as usual, by the {\schrod} equation
\begin{equation}
\label{eq:schrod-TQS}
i\hbar\frac{d\ }{d t}\ket{\psi(t)}=\wh{H}_{\tau}\ket{\psi(t)}.
\end{equation}

Its solutions for $\ket{\psi(0)}=\ket{\psi_0}$ are $\ket{\psi(t)}=\wh{U}(t)\ket{\psi_0}$, where $\wh{U}(t):=e^{-\frac{i}{\hbar}t\wh{H}_{\tau}}$ is the evolution operator.
In general,
\begin{equation}
\label{eq:time-evol}
\ket{\psi(t+\Delta t)}=e^{-\frac{i}{\hbar}\Delta t\wh{H}_{\tau}}\ket{\psi(t)}.
\end{equation}

The Hamiltonian $\wh{H}_{\tau}$ admits a canonical conjugate operator $\wh{\tau}$, so that together they satisfy the \emph{canonical commutation relation}
\begin{equation}
\label{eq:CCR-tau}
[\wh{\tau},\wh{H}_{\tau}]=i\hbar\wh{I},
\end{equation}
where $\wh{I}$ is the identity operator on $\hilbert$.

The only allowed states of the system are eigenstates of the operator $\wh{\tau}$.
The unitary evolution operator by $e^{-\frac{i}{\hbar}t\wh{H}_{\tau}}$ of an eigenstate corresponding to an eigenvalue $\tau$ results in an eigenstate corresponding to the eigenvalue $\tau+t$.
\end{defCustom}

We will see that $\wh{\tau}$ can play the role of \emph{time operator}.

\begin{remark}[Motivation for {\TQS}]
\label{rem:motivation}
There are several reasons to embark in exploring the possibility that our world is a {\TQS}.
Mainly, it seems that relativity requires that time and energy are in the same relation as position and momentum, both to ensure similar uncertainty relations, and as pairs of canonically conjugate observables.
But the \emph{Stone-von Neumann Theorem}, which (see {\eg} \cite{Hall2013QuantumTheoryForMathematicians}, p. 239)
\begin{quote}
states that any two self-adjoint operators $A$ and $B$ satisfying $[A,B] = i\hbar I$ ``look like'' [author's note: {\ie} ``it is unitarily equivalent to a direct sum of''] several copies of the standard position and momentum operators acting on $L^2(\R)$
\end{quote}
implies that, if the Hilbert space is separable, the system can only be a {\TQS}.

Another reason comes from \emph{Canonical Quantum Gravity} \cite{Dewitt1967QuantumTheoryOfGravityI_TheCanonicalTheory}, whose solutions seem to be ``timeless''. Page and Wootters proposed a formalism to accommodate this apparent timelessness with a relational notion of time based on correlations between a system and a clock \cite{PageWootters1983EvolutionWithoutEvolution}. For critical assessments of their proposal see \cite{Kuchar2011TimeAndInterpretationsOfQuantumGravity,Isham1993CanonicalQuantumGravityAndTheProblemOfTime}.
The relation between the {\TQS} and the Page-Wootters formalism will be shown in Sections \sref{s:clock} and \sref{s:time-operator-clock}.

Another reason to explore {\TQS}s is the revived interest in the Page-Wootters formalism \cite{PageWootters1983EvolutionWithoutEvolution}, this time motivated by the recent results regarding \emph{indefinite causal order}.
Indefinite causal order occurs when considering superpositions of states consisting of subsystems whose interactions take place in opposite causal orders \cite{Hardy2009QuantumGravityComputersOnTheTheoryOfComputationWithIndefiniteCausalStructure,OreshkovCostaBrukner2012QuantumCorrelationsWithNoCausalOrder,ProcopioEtAl2015ExperimentalSuperpositionOfOrdersOfQuantumGates}.
A possible and perhaps the most natural solution is to quantize time, and in particular to use the Page-Wootters formalism \cite{BaumannKrummGuerinBrukner2022NoncausalPageWoottersCircuits,SuleymanovCohen2023QuantumFramesOfReferenceAndRelationalFlowOfTime}.

A secondary reason (simply because it is unknown) is a recent result, showing that unless the world is a {\TQS}, there is always a nonvanishing probability that the world will return to a previous state \cite{Stoica2022ProblemOfIrreversibleChangeInQuantumMechanics}.
The proof from \cite{Stoica2022ProblemOfIrreversibleChangeInQuantumMechanics} is very simple: if the entire system evolves into a state that is not orthogonal to a previous state, the total state vector has a nonvanishing component corresponding to the previous state. Whether we resolve the superposition by appealing to the wavefunction collapse or to decoherence into branches, there is a nonvanishing probability that the system will return to a previous state.
To prevent this, it can be shown that, at least for the Hilbert subspace spanned by the irreversible histories, the Hamiltonian has to be the generator of some translations, so the world must be describable as a {\TQS}.
Therefore, a system that is not a {\TQS} violates the Second Law of Thermodynamics.
The following quote by Eddington gives therefore another motivation to reconsider the possibility that the world is a {\TQS} \cite{Eddington1928NatureOfThePhysicalWorld},
\begin{quote}
If your theory is found to be against the second law of thermodynamics I give you no hope; there is nothing for it but to collapse in deepest humiliation.
\end{quote}

Finally, another reason to study the {\TQS}s is that they can be very diverse and sophisticated, much like our world, as we shall see.
\end{remark}

We anticipate some possible objections against {\TQS}s, discussed in Sections \sref{s:HSF}--\sref{s:energy-extraction}:

\begin{objection}
\label{obj:simple}
The {\TQS} is unique and simple, so what can be interesting about this?
\end{objection}

\begin{objection}
\label{obj:unrealistic}
While we do not yet know the Hamiltonian of the unified theory, and hence of the entire world, we know the laws governing the Standard Model of particle physics, and the Hamiltonian has to be much more complex than \eqref{eq:time_translation_Hamiltonian}.
\end{objection}

Perhaps the best known objection, which is more general, is the following \cite{Malament1996InDefenseOfDogmaWhyCannotBeARelativisticQuantumMechanicsOfParticles}:

\begin{objection}
\label{obj:negative-energy-descent}
If the Hamiltonian of the entire world were not bounded from below, its state would decay towards infinite negative energy states.
Such a world would contain systems able to generate free energy, for example by charging infinitely many batteries.
Therefore, the spectrum of the Hamiltonian operator of the world cannot be \eqref{eq:time_translation_Hamiltonian}.
\end{objection}

We will see that there are {\TQS}s of any desired complexity, empirically indistinguishable from our world, which do not decay to infinite negative energy states, or if they decay, the same should happen to systems with the Hamiltonian bounded from below.

We will see that the following are {\TQS}s:

$\bullet$ Any closed quantum system consisting of a measuring device and an observed system realizing the \emph{standard model of ideal quantum measurements} (Section \sref{s:measurements}).
The two systems are assumed to not have free evolution.

$\bullet$ Any closed quantum system having a closed subsystem which is translational. The subsystem may be an ideal clock  (Section \sref{s:clock}) or a sterile massless fermion in a certain state (Section \sref{s:weyl}) or any other translational system.

$\bullet$ The Koopman-von Neumann-Sch\"onberg quantum representation of any deterministic time-reversible dynamical system without time loops (Section \sref{s:dyn_sys}). We will also see that all observables and properties of the original dynamical system and its evolution equations are faithfully encoded as quantum observables.

These quantum systems provide countless counterexamples of unlimited complexity to Objections \ref{obj:simple} and \ref{obj:unrealistic} (Section \sref{s:HSF}) and Objection \ref{obj:negative-energy-descent} (Section \sref{s:energy-extraction}), allowing for the possibility to interpret $\wh{\tau}$ as a time observable canonically conjugate to the Hamiltonian.
In addition, they are concrete counterexamples to the ``Hilbert Space Fundamentalism'' Thesis (\HSF) that the Hamiltonian and the state vector are sufficient to recover the $3$D-space, the tensor product structure, the preferred basis, or any other physical property, proving the necessity of taking into account the physical meaning of the observables (Section \sref{s:HSF}). The clock ambiguity problem often attributed to the Page-Wootters formalism is also due to assuming the thesis (\HSF), and it can be avoided by taking into account the physical meanings of other observables  (\sref{s:time-operator-clock}).

Finally, Section \sref{s:conclusions} summarizes the conclusions regarding the empirical adequacy of the hypothesis that the world is a {\TQS}.

\section{Preliminary remarks}
\label{s:pre-remarks}

\begin{remark}[Terminology. General]
\label{rem:terminology-translation-general}
Translation operators are also named \emph{shift operators} in Functional Analysis.
Consider a function $f(x_1,\ldots,x_j,\ldots)$ of some independent parameters $x_1,\ldots,x_j,\ldots$.
They may be space degrees of freedom, but they may be other degrees of freedom as well.
A \emph{translation operator} $\wh{T}_j\(\Delta x_j\)$ translates any function $f(x_1,\ldots,x_j,\ldots)$ with respect to the parameter $x_j$,
\begin{equation}
\label{eq:translation-op}
\wh{T}_j\(\Delta x_j\)f(x_1,\ldots,x_j,\ldots)=f(x_1,\ldots,x_j-\Delta x_j,\ldots).
\end{equation}

For example, if the parameter $x_j$ is a space coordinate and the function is a wavefunction $\psi(x_1,\ldots,x_j,\ldots)$, the translation operators along $x_j$ are generated by the corresponding momentum operator $\wh{p}_j:=-i\hbar\frac{\partial\ }{\partial x_j}$ (see Figure \ref{fig:translation-xj}),
\begin{equation}
\label{eq:translation-op-generator}
\wh{T}_j\(\Delta x_j\)\psi(x_1,\ldots,x_j,\ldots)=e^{-\frac{i}{\hbar}\Delta x_j\wh{p}_j}\psi(x_1,\ldots,x_j,\ldots)=\psi(x_1,\ldots,x_j-\Delta x_j,\ldots).
\end{equation}

\begin{figure}[!ht]
\begin{flushright}
\begin{tikzpicture}[>={Stealth[length=6pt]},declare function={g(\x)=2*exp(-\x*\x/3);
    xmax=4;xmin=-4;taumax=8;x0=1.5;ymax=2.75;}]
 \node[single arrow,left color=blue!33!white,right color=red!33!white, shift={(1.3 cm,2 cm)}]{\color{black}{$\hspace{0.55 cm}e^{-\frac{i}{\hbar}\Delta\x\cdot\wh{\p}}\hspace{0.55 cm}$}};
 \draw[black] (-4.5,0) edge[->] (taumax,0) (0,0) edge[->] (0,ymax);
 \draw[dashed, black] (-3,0) -- (-3,2.25);
 \draw[dashed, black] (3,0) -- (3,2.25);
 \draw[dashed, black] (-3.75,0.1) -- (0.75,0.1);
 \draw[dashed, black] (-0.75,2) -- (0.75,2);
 \draw[thick, blue] plot[domain=xmin:xmax,samples=51,smooth] (\x,{g(\x)}); 
 \path 
  (-3,0) node[below]{$\x_0-\Delta \x$}
  (0,0) node[below]{$\phantom{\Delta}\x_0\phantom{\Delta}$}
  (3,0) node[below]{$\x_0+\Delta \x$}
  (taumax,0) node[below]{$\x$}
  (0,2.4) node[left,blue]{$\braket{\x_0}{\psi}$}  
  (-3,0.5) node[left,blue]{$\braket{\x_0-\Delta\x}{\psi}$}
  (0.8,0.5) node[left,red]{$\bra{\x_0}e^{-\frac{i}{\hbar}\Delta\x\cdot\wh{\p}}\ket{\psi}$}
	(0,ymax) node[right]{$\psi$};
 \draw[thick, red, shift={(3 cm, 0 cm)}] plot[domain=xmin:xmax,samples=51,smooth] (\x,{g(\x)}); 
 \draw[blue,fill=blue] (-3,0.1) circle [radius=0.05];
 \draw[blue,fill=blue] (0,2) circle [radius=0.05];
 \draw[red,fill=red] (0,0.1) circle [radius=0.05];
 \draw[red,fill=red] (3,2) circle [radius=0.05];
\end{tikzpicture}
\caption{The operator $e^{-\frac{i}{\hbar}\Delta\x\cdot\wh{\p}}$ translates $\psi$ with $\Delta\x$.}
\label{fig:translation-xj}
\end{flushright}
\end{figure}

If a function $\psi(x_1,\ldots,x_j,\ldots)$ is translated with $\Delta x_j$, in the original reference frame the value of the translated function at $x_j=0$ is the value of the original function at $x_j-\Delta x_j$, as in equation \eqref{rem:terminology-translation-general}.
This is akin to a (passive) change of the reference frame, so an eigenvector $\ket{\wt{x}_j}$ of the position operator $\wh{x}_j$ is transformed into the eigenvector $\ket{\wt{x}_j+\Delta x_j}$ (we denote the eigenvectors by the corresponding eigenvalues). We can check in one dimension that indeed $\braket{x}{\wt{x}+\Delta x}=\delta\(x-(\wt{x}+\Delta x)\)=\delta\((x-\wt{x})-\Delta x\)=\wh{T}\(\Delta x\)\delta(x-\wt{x})=\bra{x}\wh{T}\(\Delta x\)\ket{\wt{x}}$.
This is not different from the usual changes of reference frames: the coordinates in the new reference frame are obtained by applying to the coordinates in the old reference frame the inverse of the transformation changing the old reference frame into the new one.
\end{remark}

Let $\ket{\wt{\tau},\A}$ be an eigenstate of $\wh{\tau}$ corresponding to the eigenvalue $\wt{\tau}$, where we denote by $\A=\(a_1,a_2,\ldots\)$ the additional parameters needed to account for the degeneracy of the eigenvalue $\tau$. They are eigenvalues of some Hermitian operators $\wh{\A}=\(\wh{a}_1,\wh{a}_2,\ldots\)$ that can be chosen to form, together with $\wh{\tau}$, a complete set of commuting observables, so they are simultaneously diagonalizable.

There are many ways to choose the operators $\wh{a}_1,\wh{a}_2,\ldots$, resulting in many different representations. The following choice makes explicit how $\wh{H}_{\tau}$ generates translations.

\begin{example}[Parmenidean representation]
\label{ex:parmenidean}
We can choose $\wh{a}_1,\wh{a}_2,\ldots$ to also commute with $\wh{H}_{\tau}$, so that they are conserved.
This is possible because the Hilbert space $\hilbert$ is isomorphic with the tensor product of Hilbert spaces, $\hilbert\cong\hilbert_{C}\otimes\hilbert_{R}$, where $\hilbert_{C}:=L^2\(\R,\C\)$ is the space of square-integrable functions and $\hilbert_{R}$ is another Hilbert space, such that the Hamiltonian has the form $\wh{H}_{\tau}=\wh{H}_{C}\otimes\wh{I}_{R}$, where $\wh{H}_{C}=-i\hbar\frac{\partial\ }{\partial\tau_{C}}$, $\tau_{C}\in\R$, and $\wh{I}_{R}$ is the identity operator on $\hilbert_{R}$.
We choose the operators $\wh{a}_1,\wh{a}_2,\ldots$ of the form $\wh{a}_j:=\wh{I}_{C}\otimes\wh{b}_j$, where $\wh{b}_1,\wh{b}_2,\ldots$ form a complete set of commuting observables on $\hilbert_{R}$.
Then, the action of $\wh{U}(t)$ on $\ket{\wt{\tau},\A}$ changes it into an eigenstate $\ket{\wt{\tau}+t,\A}$ corresponding to the eigenvalue $\wt{\tau}+t$.
Therefore, the solution of equation \eqref{eq:schrod-TQS} with the initial condition $\ket{\psi_0}=\ket{\tau,\A}$ is $\ket{\psi(t)}=\ket{\tau+t,\A}$,
\begin{equation}
\label{eq:time-translation-eigentau-parmenidean}
\ket{\tau+t,\A}=e^{-\frac{i}{\hbar}t\wh{H}_{\tau}}\ket{\tau,\A}.
\end{equation}

A state vector $\ket{\psi(0)}$ has the general form
\begin{equation}
\label{eq:time-translation-general-psi}
\ket{\psi(0)}=\int_{\mc{C}} \psi_0(\tau,\A)\ket{\tau,\A}\de\tau\de\A,
\end{equation}
where $\psi_0(\tau,\A)=\braket{\tau,\A}{\psi(0)}$ is the wavefunction and $\mc{C}$ is the configuration space, parametrized by $\tau$ and $\A$.
The state vector is evolved by $e^{-\frac{i}{\hbar}t\wh{H}_{\tau}}$ into
\begin{equation}
\label{eq:time-translation-psi}
\begin{aligned}
\ket{\psi(t)}
&=e^{-\frac{i}{\hbar}t\wh{H}_{\tau}}\ket{\psi(0)} \\
&=\int_{\R^{n+1}} \psi_0(\tau,\A)\ket{\tau+t,\A}\de\tau\de\A \\
&=\int_{\R^{n+1}} \psi_0(\tau-t,\A)\ket{\tau,\A}\de\tau\de\A,
\end{aligned}
\end{equation}
where $n$ is the number of parameters $\A=\(a_1,a_2,\ldots\)$, or, equivalently, of the operators $\wh{a}_1,\wh{a}_2,\ldots$ forming together with $\wh{\tau}$ a complete set of commuting observables on $\hilbert$.
The number $n$ is the degeneracy (or multiplicity) of the operator $\wh{\tau}$.

Then, 
\begin{equation}
\label{eq:time-translation}
\psi_0(\tau-t,\A)=e^{-\frac{i}{\hbar}t\wh{H}_{\tau}}\psi_0(\tau,\A).
\end{equation}

Therefore, a wavefunction $\psi_0(\tau,\A)$ is evolved by $e^{-\frac{i}{\hbar}t\wh{H}_{\tau}}$ into $\psi_t(\tau,\A)=\psi_0(\tau-t,\A)$.
For example, for $t>0$, $\psi_0(0,\A)$ is translated forward along $\tau$ (see Figure \ref{fig:translation-tau}).

\begin{figure}[!ht]
\begin{flushright}
\begin{tikzpicture}[>={Stealth[length=6pt]},declare function={g(\x)=2*exp(-\x*\x/3);
    xmax=4;xmin=-4;taumax=8;x0=1.5;ymax=2.75;}]
 \node[single arrow,left color=blue!33!white,right color=red!33!white, shift={(1.3 cm,2 cm)}]{\color{black}{$\hspace{0.55 cm}e^{-\frac{i}{\hbar}t\wh{H}_{\tau}}\hspace{0.55 cm}$}};
 \draw[black] (-4.5,0) edge[->] (taumax,0) (0,0) edge[->] (0,ymax);
 \draw[dashed, black] (-3,0) -- (-3,2.25);
 \draw[dashed, black] (3,0) -- (3,2.25);
 \draw[dashed, black] (-3.75,0.1) -- (0.75,0.1);
 \draw[dashed, black] (-0.75,2) -- (0.75,2);
 \draw[thick, blue] plot[domain=xmin:xmax,samples=51,smooth] (\x,{g(\x)}); 
 \path 
  (-3,0) node[below]{$\tau=-t$}
  (0,0) node[below]{$\tau=0$}
  (3,0) node[below]{$\tau=t$}
  (taumax,0) node[below]{$\tau$}
  (0,2.25) node[left,blue]{$\psi_0(0,\A)$}  
  (-3,0.35) node[left,blue]{$\psi_0(-t,\A)$}  
  (0,0.35) node[left,red]{$\psi_t(0,\A)$}
	(0,ymax) node[right]{$\psi$};
 \draw[thick, red, shift={(3 cm, 0 cm)}] plot[domain=xmin:xmax,samples=51,smooth] (\x,{g(\x)}); 
 \draw[blue,fill=blue] (-3,0.1) circle [radius=0.05];
 \draw[blue,fill=blue] (0,2) circle [radius=0.05];
 \draw[red,fill=red] (0,0.1) circle [radius=0.05];
 \draw[red,fill=red] (3,2) circle [radius=0.05];
\end{tikzpicture}
\caption{The operator $e^{-\frac{i}{\hbar}t\wh{H}_{\tau}}$ translates $\psi_0$ forward with $t$ along the parameter $\tau$.}
\label{fig:translation-tau}
\end{flushright}
\end{figure}

Since in this representation there is no explicit change under time evolution (except for the value of $\tau$), we can call it \emph{Parmenidean representation}.
It is similar to a physical system that does not change in time, even if it undergoes inertial motion.
However, in Section \sref{s:clock}, Remark \ref{rem:PW}, we will recall that Canonical Quantum Gravity implies that there is a ``truly Parmenidean'' representation that appears to be timeless.
\end{example}

This may give the impression that translational quantum systems are unable to describe any interesting physics.
The next example shows that this is not the case, and is the key to the examples from this article.

\begin{example}[Heraclitean representation]
\label{ex:heraclitean}
Suppose now that the operators $\wh{a}_1,\wh{a}_2,\ldots$, commuting with $\wh{\tau}$, are chosen so that not all of them commute with $\wh{H}_{\tau}$.
Then the unitary evolution of a state represented by a state vector $\ket{\psi(0)}$, which is initially an eigenstate $\ket{0}_{C}\ket{\varphi(0)}_{R}$ of $\wh{\tau}$, where  $\ket{0}_{C}\in\hilbert_{C}$ and $\ket{\varphi(0)}_{R}\in\hilbert_{R}$, has the general form
\begin{equation}
\label{eq:time-translation-eigentau-heraclitean}
\ket{t}_{C}\ket{\varphi(t)}_{R}=e^{-\frac{i}{\hbar}t\wh{H}_{\tau}}\ket{0}_{C}\ket{\varphi(0)}_{R},
\end{equation}
so that at any time the system is in an eigenstate of $\wh{\tau}$, but $\ket{\varphi(t)}_{R}$ changes in time.
The subsystem $\(\hilbert_{C},\wh{H}_{\tau_{C}}\)$ plays the role of a \emph{clock}. It will be discussed in Section \sref{s:clock}, where we will see that it is possible for the rest of the system to have any Hamiltonian.

The possibility that the operators $\wh{a}_1,\wh{a}_2,\ldots$ don't commute with the Hamiltonian, together with the possibility to interpret them physically in many different ways, allows translational quantum systems to be very diverse, proving their versatility. Sections \sref{s:measurements}--\sref{s:dyn_sys} illustrate this versatility with concrete examples.
\end{example}

\begin{remark}[Terminology. Time translation]
\label{rem:terminology-translation-time}
The unitary evolution operators $e^{-\frac{i}{\hbar}t\wh{H}}$ of any closed quantum system are sometimes called \emph{time translations} (\eg \cite{AulettaFortunatoParisi2009QuantumMechanics} p. 101).
But, considering that any unitary operator can be put in an exponential form $e^{-\frac{i}{\hbar}t\wh{H}}$ for some Hermitian operator $\wh{H}$, does this mean that all unitary operators are translations?
The answer is no, the notion of ``time translation operator'' is different from the notion of ``translation operator''.
Translation operators like $e^{-\frac{i}{\hbar}\Delta x_j\wh{p}_j}$ and $e^{-\frac{i}{\hbar}\Delta\tau\wh{H}_{\tau}}$ produce translations with respect to the parameters $x_j$, respectively $\tau$, which are degrees of freedom of the system. But the time parameter is not a degree of freedom of the system. For example, if $\ket{\psi(t_1)}=\ket{\psi(t_2)}$, the two states are identical with respect to all degrees of freedom of the system even if $t_1\neq t_2$, while for a {\TQS} different values of $\tau$ correspond to orthogonal states.

However, it is possible to interpret the unitary evolution operators of a generic closed quantum system as translations, but only in a limited sense.
To do this, we can interpret the time-dependent wavefunction $\psi(\q,t)$, where $\q\in\mc{C}$ is a point in the configuration space $\mc{C}$ of the system, as a wavefunction on a \emph{time-extended} configuration space $\mc{C}\times\R$.
Here $\mc{C}\times\R$ denotes the Cartesian product between the original configuration space of the system and the time axis $\R$.
This makes time a degree of freedom of the system.
In this case,
\begin{equation}
\label{eq:square-integral-time-extended-config-space}
\int_{\mc{C}\times\R}\abs{\psi(\q,t)}^2\de\q d t=\infty,
\end{equation}
so the time-extended wavefunctions do not form a Hilbert space.
But for any fixed $t\in\R$, $\int_{\mc{C}}\abs{\psi(\q,t)}^2\de\q=1$ as usual, so the wavefunctions $\psi(\q,t)$ defined on $\mc{C}\times\{t\}$ form a Hilbert space $\hilbert_t$.

The unitary evolution operators $e^{-\frac{i}{\hbar}t\wh{H}}$ act like translations on the time-extended configuration space $\mc{C}\times\R$, but only for those time-extended wavefunctions that are solutions of the {\schrod} equation for the Hamiltonian $\wh{H}$.
If $\psi(\q,t)=\braket{\q}{\psi(t)}$ is a general wavefunction, or if it is a solution of the {\schrod} equation for a different Hamiltonian $\wh{H}'$, obviously $\wh{H}$ does not generate the time translation of $\psi(\q,t)$, because in general $e^{-\frac{i}{\hbar} t\wh{H}}\ket{\psi(t)}\neq e^{-\frac{i}{\hbar} t\wh{H}'}\ket{\psi(t)}$.

This is true in particular for the time-extended configuration space of a {\TQS}, obtained by including the time among the degrees of freedom of the system: the unitary evolution operators $e^{-\frac{i}{\hbar}t\wh{H}_{\tau}}$ act like time translations only on the time-extended wavefunctions that are solutions of the {\schrod} equation \eqref{eq:schrod-TQS}.

On the other hand, on the original configuration space, the translation operators $e^{-\frac{i}{\hbar}t\wh{H}_{\tau}}$ translate any function $\psi(\tau,\A)$ with respect to $\tau$, as in equation \eqref{eq:time-translation}.
This applies even to those functions that don't satisfy the {\schrod} equation with the Hamiltonian \eqref{eq:time_translation_Hamiltonian}, so the operators $e^{-\frac{i}{\hbar}t\wh{H}_{\tau}}$ satisfy the standard definition of translation operators.
\end{remark}

\begin{remark}[Relation between $\tau$ and $t$]
\label{rem:tau-and-t}
It may seem that, by working on the time-extended configuration space as in Remark \ref{rem:terminology-translation-time}, we can introduce a time observable for any Hamiltonian.
Even if we find a way to avoid the technical difficulties arising from the fact that the time-extended wavefunctions do not form a Hilbert space, as seen in equation \eqref{eq:square-integral-time-extended-config-space}, such a time operator would not satisfy the canonical commutation relation with the Hamiltonian.

On the other hand, in the case of a {\TQS}, $\tau$ is a degree of freedom of the system, and since the allowed solutions are eigenstates of $\wh{\tau}$, $\wh{\tau}$ can play the role of a time operator, and it satisfies the canonical commutation relation \eqref{eq:CCR-tau}, just like we would expect for time and energy.

From Definition \ref{def:TQS}, the allowed states of a {\TQS} are eigenstates of the operator $\wh{\tau}$, and they remain eigenstatesof $\wh{\tau}$ at any time $t$. Suppose that at time $t_0$ the state of the system is in an eigenstate of $\wh{\tau}$, corresponding to the eigenvalue $\tau_0$. Then, at any other time $t_0+\Delta t$, the state of the {\TQS} is $\tau_0+\Delta t$, {\cf} equations \eqref{eq:time-translation-eigentau-parmenidean} and \eqref{eq:time-translation-eigentau-heraclitean}.
Therefore, $\tau$ and $t$ are correlated so that as $t$ increases with $\Delta t$, $\tau$ increases with the same amount.
This justifies the interpretation of $\tau$ as time.
\end{remark}

\begin{remark}[Probabilities for a {\TQS}]
\label{rem:interpretation-probability}
A {\TQS} is a quantum system like any other quantum system, with some particularities.
One of these particularities is that the state vector is an eigenstate of the operator $\wh{\tau}$, and it remains an eigenstate at all times.
For example, {\cf} equation \eqref{eq:time-translation-eigentau-parmenidean}, at any time $t$ the state vector is $\ket{t,\A}=e^{-\frac{i}{\hbar}t\wh{H}_{\tau}}\ket{0,\A}$.
While $\ket{\tau,\A}$ is not normalizable because $\braket{\tau_1,\A}{\tau_2,\A}$ is a Dirac distribution, it satisfies the ``normalization condition'' $\braket{\tau_1,\A}{\tau_2,\A}=\delta\(\tau_1-\tau_2\)$.
Unitary evolution ensures that this remains true at all times.

Another particularity is that the Hamiltonian \eqref{eq:time_translation_Hamiltonian} admits a canonical conjugate $\wh{\tau}$.
If we do not know the value $\tau$, we represent this limited knowledge by a state vector
\begin{equation}
\label{eq:square-integral-time-extended-config-space-tau}
\ket{\psi}=\int_{\mc{C}} \psi(\tau,\A)\ket{\tau,\A}\de\tau\de\A,
\end{equation}
as in equation \eqref{eq:time-translation-general-psi}, where $\braket{\psi}{\psi}=1$.
Then, $\braket{\tau}{\psi}$ is the probability distribution for the value of $\tau$, just like in the case of any observable, for example the position observable.
And since {\cf} Remark \ref{rem:tau-and-t} the value $\tau$ correlates with time,  $\braket{\tau}{\psi}$ can be interpreted as the probability distribution for time.

Note that we do not consider the state of the {\TQS} to be represented by the totality of states at all times, $\int_{\mc{C}}\ket{\tau,\A}\de\tau\de\A$ as for the time-extended wavefunction \eqref{eq:square-integral-time-extended-config-space}.
There is no reason to do this, and no probabilistic interpretation for this.
Therefore, we do not have to deal with the fact that $\int_{\R}\braket{\tau,\A}{\tau,\A}d \tau=\infty$. This corresponds to complete lack of knowledge about $\tau$, similar to how the uncertainty of position is maximal for the eigenstates of the momentum observable.
Also, $\tau$ cannot stand for the absolute time, but it can be interpreted as relative time, just like $\x$ represents relative position.
\end{remark}

\begin{remark}[Sign convention]
\label{rem:sign-convention}
From Remark \ref{rem:terminology-translation-time} we learn that, even though any Hamiltonian operator generates time translations for the solutions of the corresponding {\schrod} equation, in order to generate translations for all wavefunctions on the configuration space, the Hamiltonian has to be $-i\hbar\frac{\partial\ }{\partial\tau}$ (or a multiple $-ir\hbar\frac{\partial\ }{\partial\tau}$ of it, where $r\neq 0$ is a real constant).
Since $r\neq 0$, it can be absorbed in $\tau$, by substituting $\tau$ with $\tau/r$. But even so, there are two distinct interesting choices for the Hamiltonian, $-i\hbar\frac{\partial\ }{\partial\tau}$ and $i\hbar\frac{\partial\ }{\partial\tau}$, and both are used in the literature. 
The choice $\wh{H}_{\tau}=i\hbar\frac{\partial\ }{\partial\tau}$ satisfies the canonical commutation relation $[\wh{\tau},\wh{H}_{\tau}]=-i\hbar\wh{I}$.
The choice $\wh{H}_{\tau}=-i\hbar\frac{\partial\ }{\partial\tau}$ from equation \eqref{eq:time_translation_Hamiltonian} satisfies the canonical commutation relation $[\wh{\tau},\wh{H}_{\tau}]=i\hbar\wh{I}$.
Which choice is better may depend of applications and personal preferences. For example, in \cite{SrinivasVijayalakshmi1981TimeOfOccurrenceInQuantumMechanics}, p. 183, it is argued for $[\wh{\tau},\wh{H}_{\tau}]=-i\hbar\wh{I}$, because this ``is the requirement of time-translation invariance of the theory'', while Jammer argues that $[\wh{\tau},\wh{H}_{\tau}]=i\hbar\wh{I}$ is the right choice, because ``the time-energy and position-momentum relations would have the same logical status'' and ``the relativity requirement of treating time and position coordinates as well as energy and momentum components on a common footing'' (\cite{Jammer1974ThePhilosophyOfQuantumMechanics} p. 150).

If we compare equations \eqref{eq:time-translation} and \eqref{eq:time-evol}, we see that the time evolution generated by the Hamiltonian $-i\hbar\frac{\partial\ }{\partial\tau}$ translates the wavefunction along $\tau$ in opposite direction to how it translates it along $t$. So, at first sight, it may seem better to choose $i\hbar\frac{\partial\ }{\partial\tau}$ as the Hamiltonian.
But there is a difference between time as a parameter that is not a degree of freedom of the system, and $\tau$, which is, and this is why the translation happens differently with respect to $t$ compared to $\tau$.
With the choice $-i\hbar\frac{\partial\ }{\partial\tau}$, the evolution operator $e^{-\frac{i}{\hbar}t\wh{H}_{\tau}}$ translates an eigenstate $\ket{\tau}_{C}\ket{\varphi(\tau)}_{R}$ to an eigenstate $\ket{\tau+t}_{C}\ket{\varphi(\tau+t)}_{R}$, so that the corresponding eigenvalue increases with $t$, as in equation \eqref{eq:time-translation-eigentau-heraclitean}.
This motivates equations \eqref{eq:time_translation_Hamiltonian} and \eqref{eq:CCR-tau} as the appropriate choice if we want to treat time as a degree of freedom of the system.
\end{remark}

\section{Quantum measurements}
\label{s:measurements}

In this Section I show that the Hamiltonian of a well known model of ideal measurements, named the \emph{standard model of measurements} in \cite{BuschGrabowskiLahti1995OperationalQuantumPhysics}  \S II.3.4, is a {\TQS}.

Let $M$ be the measuring device, $S$ the observed system, and $\hilbert_{M}$ and $\hilbert_{S}$ their Hilbert spaces.
Suppose that a property of system $S$ is measured, represented by the Hermitian operator $\wh{\mc{O}}$ on $\hilbert_{S}$ with discrete spectrum $\sigma(\wh{\mc{O}})$ and eigenstates $\ket{\lambda,a}_{S}$, $\lambda\in\sigma(\wh{\mc{O}})$, where $a$ is a degeneracy index.
We assume that $0$ is not an eigenvalue of $\wh{\mc{O}}$.

We assume that the pointer observable is a Hermitian operator $\wh{\mc{M}}$ on $\hilbert_{M}$, with nondegenerate and continuous spectrum $\sigma(\wh{\mc{M}})$, and eigenstates $\ket{\zeta}_{M}$ labeled by $\zeta\in\sigma(\wh{\mc{M}})$.
Since the pointer states have to represent the eigenstates of the operator $\wh{\mc{O}}$, we assume that an injective function $\zeta:\sigma(\wh{\mc{O}})\to\sigma(\wh{\mc{M}})$ associates the eigenvalue $\zeta=\zeta(\lambda)$ to each $\lambda\in\wh{\mc{O}}$, so that $\zeta(\lambda)\neq 0$ for all $\lambda\in\sigma(\wh{\mc{O}})$. We reserve $0\in\sigma(\wh{\mc{M}})$ to represent the ``ready'' pointer state.

Suppose that the measurement starts at $t=0$ and ends at $t=T>0$.
Let $\wh{U}$ be the unitary evolution operator of the total system.
For each $\lambda\in\sigma(\wh{\mc{O}})$, we assume that
\begin{equation}
\label{eq:calibration}
\wh{U}(T)\ket{\lambda,a}_{S}\ket{\tn{ready}}_{M} = \ket{\lambda,a}_{S}\ket{\zeta(\lambda)}_{M}.
\end{equation}

For any initial state of the observed system, the total system evolves into a superposition of states \eqref{eq:calibration}.
Then, the Projection Postulate (or another mechanism based on decoherence) is invoked to explain that only one term in the superposition remains.
If the pointer's final state is $\ket{\zeta(\lambda)}_{M}$, the result of the measurement is taken to be the eigenvalue $\lambda$ of $\mc{O}$.

In the standard model of quantum measurements (see \eg \cite{Mittelstaedt2004InterpretationOfQMAndMeasurementProcess} \S 2.2(b) and \cite{BuschGrabowskiLahti1995OperationalQuantumPhysics} \S II.3.4), condition \eqref{eq:calibration} is ensured by the following Hamiltonian, which ignores the free Hamiltonians of the two systems or assumes them to vanish,
\begin{equation}
\label{eq:measurement_hamiltonian}
\wh{H}=\wh{H}_{\tn{int}} = g \wh{\mc{O}}\otimes\wh{p}_{\mc{M}},
\end{equation}
$\wh{p}_{\mc{M}}$ is the canonical conjugate of the pointer operator $\wh{\mc{M}}$ with $\sigma(\wh{\mc{M}})=\R$.
We assume that the coupling $g$ is constant in the interval $\left[0,T\right]$ and negligible outside, for example due to an interaction potential.
Then, for any $\lambda$ and $a$,
\begin{equation}
\label{eq:interaction_unitary_evolution_operator_action}
\wh{U}(T)\ket{\lambda,a}_{S}\ket{\tn{ready}}_{M} = \ket{\lambda,a}_{S} e^{-\frac{i}{\hbar}g T \lambda \wh{p}_{\mc{M}}}\ket{\tn{ready}}_{M}.
\end{equation}

Since $[\wh{\mc{M}},\wh{p}_{\mc{M}}]=i\hbar\wh{I}$, we get for any $\lambda$
\begin{equation}
\label{eq:interaction_unitary_evolution_operator_transform}
e^{ig T \hbar^{-1} \lambda \wh{p}_{\mc{M}}}\wh{\mc{M}}e^{-ig T \hbar^{-1} \lambda \wh{p}_{\mc{M}}} = \wh{\mc{M}} + g T \hbar^{-1} \lambda \wh{1}_{M}.
\end{equation}

From \eqref{eq:calibration} and \eqref{eq:interaction_unitary_evolution_operator_action}, the pointer eigenstates are
\begin{equation}
\label{eq:pointer_eigenstates}
\ket{\zeta(\lambda)}_{M} = e^{-\frac{i}{\hbar}g T \lambda \wh{p}_{\mc{M}}}\ket{\tn{ready}}_{M}.
\end{equation}

From \eqref{eq:interaction_unitary_evolution_operator_transform} and $\zeta_\tn{ready}=0$ the pointer eigenvalues are
\begin{equation}
\label{eq:pointer_eigenvalues}
\zeta(\lambda)= g T\lambda.
\end{equation}

Now I prove the main result of this Section.

\begin{theorem}
\label{thm:measurement}
There is a basis of the total Hilbert space $\hilbert=\hilbert_{M}\otimes\hilbert_{S}$ in which the Hamiltonian from eq. \eqref{eq:measurement_hamiltonian} has the form $\wh{H} = -i\hbar\frac{\partial\ }{\partial\tau}$ on a subspace containing the possible measurement results.
\end{theorem}
\begin{proof}
The representation of the pointer observable $\wh{\mc{M}}$ is $\bra{\zeta}\wh{\mc{M}}\ket{\psi}_{M} = \zeta\braket{\zeta}{\psi}_{M}$, and that of $\wh{p}_{\mc{M}}$ is $\wh{p}_{\mc{M}}=-i\hbar\frac{\partial\ }{\partial\zeta}$.
The Hamiltonian \eqref{eq:measurement_hamiltonian} becomes
\begin{equation}
\label{eq:interaction_hamiltonian_xp}
\wh{H}=-g \sum_{\lambda,a}\lambda\ket{\lambda,a}\bra{\lambda,a}\otimes i\hbar\frac{\partial\ }{\partial\zeta}.
\end{equation}

On the total Hilbert space we define an operator $\wh{\tau}$ whose restriction on each subspace $\ket{\lambda,a}_{S}\otimes\hilbert_{M}$ is
\begin{equation}
\label{eq:operator_x}
 \wh{\tau}|_{\ket{\lambda,a}_{S}\otimes\hilbert_{M}}=\frac{1}{g \lambda}\wh{\mc{M}}.
\end{equation}

On each subspace $\{\ket{\lambda,a}_{S}\}\otimes\hilbert_{M}$, where $\{\ket{\lambda,a}_{S}\}$ is the set containing only the state vector $\ket{\lambda,a}_{S}$, its eigenvalue $\tau$ satisfies
\begin{equation}
\label{eq:operator_x_eigenvalues}
\zeta=g \lambda \tau.
\end{equation}

Then, on each subspace $\{\ket{\lambda,a}_{S}\}\otimes\hilbert_{M}$,
\begin{equation}
\label{eq:operator_partial_xi}
-g\lambda i\hbar\frac{\partial}{\partial \zeta}
=-g\lambda i\hbar\frac{\partial}{\partial (g \lambda \tau)}
=-i\hbar\frac{\partial\ }{\partial\tau}.
\end{equation}

Therefore, $\wh{H} = -i\hbar\frac{\partial\ }{\partial\tau}$ on the subspace spanned by the state vectors of the form $\wh{U}(t)\ket{\lambda,a}_{S}\ket{\tn{ready}}_{M}$, for all $\lambda\in\sigma(\wh{\mc{O}})$, all $a$, and all $t\in[0,T]$.
If $g$ is constant for all times, this extends to the entire Hilbert space $\hilbert_{S}\otimes\hilbert_{M}$.
\end{proof}

\begin{remark}
\label{rem:thm:measurement:factorization}
The unitary transformation from the proof of Theorem \ref{thm:measurement} does not preserve the factorization into subsystems $\hilbert_{M}\otimes\hilbert_{S}$.
It is not a change of basis of the factor Hilbert spaces $\hilbert_{M}$ and $\hilbert_{S}$ separately, but of the total Hilbert space $\hilbert_{M}\otimes\hilbert_{S}$, and this is consistent with the statement of Theorem \ref{thm:measurement}.
\end{remark}

\begin{remark}
\label{rem:thm:measurement:coupling}
In practice, the coupling $g$ from eq. \eqref{eq:measurement_hamiltonian} is negligible outside the time interval $\left[0,T\right]$ when the measuring device and the observed system interact in a non-negligible way. Theorem \ref{thm:measurement} applies only to this interval, and only if the total system $S+M$ is closed. Whether the Hamiltonian can be put in the form \eqref{eq:time_translation_Hamiltonian} for all times, including before and after the measurement, depends on the mechanism of turning the interaction on and off. In most practical cases this is achieved by having $g$ depend on the distance between $S$ and $M$. This can happen if the position of the observed system changes in time, so its free Hamiltonian is not zero.
Alternatively, $g$ can depend on an external system responsible for turning the coupling on and off, which would make the total Hamiltonian time dependent. 
Theorem \ref{thm:measurement} does not cover these cases.
\end{remark}

\begin{remark}
\label{rem:thm:measurement:ideal}
The standard model of quantum measurements described in \cite{BuschGrabowskiLahti1995OperationalQuantumPhysics} \S II.3.4 applies both to ideal (projective) and generalized (POVM) measurements.
There are many concrete realizations of this model, including for the Stern-Gerlach experiment and various experiments with photons (\cite{BuschGrabowskiLahti1995OperationalQuantumPhysics} \S VII). 
But in general these measurements are unsharp in practice, hence the necessity of POVM, while Theorem \ref{thm:measurement} applies only to ideal projective measurements.
\end{remark}

\section{Systems containing ideal clocks}
\label{s:clock}

In this Section I show that any quantum system containing an ideal clock, or, more generally, any closed translational subsystem, has the Hamiltonian \eqref{eq:time_translation_Hamiltonian}.

Consider a quantum system composed of two closed subsystems: a generic system $R$ with state vectors $\ket{\psi(t)}\in\hilbert_{R}$ and any Hamiltonian $\wh{H}_{R}$, and a {\TQS} $C$ with state vectors $\ket{\eta(t)}\in\hilbert_{C}$ and Hamiltonian $\wh{H}_{C}=-i\hbar\frac{\partial\ }{\partial\tau_{C}}$.
Let $\wh{\tau}_{C}$ be the canonical conjugate of $-i\hbar\frac{\partial\ }{\partial\tau_{C}}$.
At $t=0$, and therefore at any time $t$, $\ket{\eta(t)}$ is assumed to be an eigenstate of $\wh{\tau}_{C}$.
The total Hamiltonian is
\begin{equation}
\label{eq:CR}
\wh{H}_{C+R}=\wh{H}_{C}\otimes \wh{I}_{R}+\wh{I}_{C}\otimes\wh{H}_{R}.
\end{equation}

\begin{theorem}
\label{thm:clock}
There is a basis of the total Hilbert space $\hilbert_{C+R}:=\hilbert_{C}\otimes\hilbert_{R}$, consisting of eigenstates of $\wh{\tau}:=\wh{\tau}_{C}\otimes \wh{I}_{R}$, in which
$\wh{H}_{C+R} = -i\hbar\frac{\partial\ }{\partial\tau}$.
The total Hamiltonian $\wh{H}_{C+R}$ is the canonical conjugate of the operator $\wh{\tau}$.
\end{theorem}
\begin{proof}
Let $\wh{U}_{C+R}(t) :=e^{-\frac{i}{\hbar}t\wh{H}_{C+R}}$.
Let $\(\ket{b}\)_{b\in\mc{B}}$ be an orthonormal basis of $\hilbert_{R}$.

Then, for any $\tau_1\neq \tau_2\in\R$, any eigenstate $\ket{\eta(\tau_1)}\in\hilbert_{C}$ of $\wh{\tau}_{C}$, and any $b\in\mc{B}$,
\begin{equation}
\label{eq:unitary_evolution_CR}
\wh{U}_{C+R}(\tau_2-\tau_1)\ket{\eta(\tau_1)}\ket{b}
=\ket{\eta(\tau_2)}e^{-\frac{i}{\hbar}t\wh{H}_{R}}(\tau_2-\tau_1)\ket{b}.
\end{equation}

Since $\braket{\eta(\tau_2)}{\eta(\tau_1)}=0$, from eq. \eqref{eq:unitary_evolution_CR} it follows that $\wh{U}_{C+R}(\tau_1)\ket{\eta(0)}\ket{b}$ and $\wh{U}_{C+R}(\tau_2)\ket{\eta(0)}\ket{b}$ are orthogonal for any $\tau_1\neq \tau_2$.
Therefore, the system of vectors $\{\wh{U}_{C+R}(\tau)\ket{\eta(0)}\ket{b}|\tau\in\R\}$ is orthogonal, and the one-parameter unitary group $\(\wh{U}_{C+R}(\tau)\)_{\tau\in\R}$ acts on it like a translation group.

We recall the following theorem, due to Stone \cite{Stone1932OnOneParameterUnitaryGroupsInHilbertSpace} (\cite{Hall2013QuantumTheoryForMathematicians}, Theorem 10.15 p. 210):

\begin{proof}[Stone's Theorem]
If $\(\wh{U}(t)\)_{t\in\R}$ is a strongly continuous one-parameter unitary group on the Hilbert space $\hilbert$, there is a unique densely defined self-adjoint operator $\wh{A}$ which is the infinitesimal generator of $\(\wh{U}(t)\)_{t\in\R}$, and $\wh{U}(t) = e^{it\wh{A}}$ for all $t\in\R$.
\end{proof}

From the Stone Theorem we obtain that
the restriction of $\wh{H}_{C+R}$ to the subspace
\begin{equation}
\label{eq:CR_span}
\mc{V}_{b}:=\Span\{\wh{U}_{C+R}(\tau)\ket{\eta(0)}\ket{b}|\tau\in\R\}
\end{equation}
is unique, and therefore it is $-i\hbar\frac{\partial\ }{\partial\tau}$.

For any $b_1\neq b_2\in\mc{B}$ and $\tau_1,\tau_2\in\R$, the state vectors $\wh{U}_{C+R}(\tau_1)\ket{\eta(0)}\ket{b_1}$ and $\wh{U}_{C+R}(\tau_2)\ket{\eta(0)}\ket{b_2}$ are orthogonal, because at equal times $e^{-\frac{i}{\hbar}t\wh{H}_{R}}\ket{b_1}$ and $e^{-\frac{i}{\hbar}t\wh{H}_{R}}\ket{b_2}$ are orthogonal and $\braket{\eta(\tau_1)}{\eta(\tau_2)}=\delta\(\tau_1-\tau_2\)$.
Therefore, for any $b_1\neq b_2\in\mc{B}$, $\mc{V}_{b_1}\perp\mc{V}_{b_2}$. It follows that the total Hilbert space is the direct sum of the subspaces $\mc{V}_{b}$, $\hilbert_{C+R}=\boplus_{b\in\mc{B}}\mc{V}_{b}$, and $\wh{H}_{C+R}$ has the form $-i\hbar\frac{\partial\ }{\partial\tau}$.
Its canonical conjugate is the operator $\wh{\tau}=\wh{\tau}_{C}\otimes \wh{I}_{R}$.
\end{proof}

In particular, system $C$ may be an ideal clock.

\begin{remark}
\label{rem:clock-versatility}
Theorem \ref{thm:clock} shows that any closed quantum system combined with a {\TQS} has the Hamiltonian \eqref{eq:time_translation_Hamiltonian}. But the dynamics of the original system does not change, it simply becomes part of the translation in the new basis. When expressed in function of observables, for example in function of positions and momenta, it is still there in its original form.
\end{remark}

\begin{remark}
\label{rem:clock-translation}
Curiously, by its mere existence, a translational subsystem makes any quantum system $C+R$ translational.
However, in general, the Hamiltonian \eqref{eq:time_translation_Hamiltonian} does not require the existence of a subsystem that explicitly acts like a clock, but it is always possible to exhibit a clock, as seen in the representations from Examples \ref{ex:parmenidean} and \ref{ex:heraclitean}.
\end{remark}

\begin{remark}
\label{rem:PW}
For any state vector $\ket{\psi(0)}_{R}$, we can represent the orbit $\{\ket{t}_{C}\ket{\psi(t)}_{R}|t\in\R\}$ of $\ket{0}_{C}\ket{\psi(0)}_{R}$ under the action of the group $\(\wh{U}_{C+R}(t,0)\)_{t\in\R}$ as a vector
\begin{equation}
\label{eq:timeless}
\kett{\Psi}=\int_{\R}\ket{\tau}_{C}\ket{\psi(\tau)}_{R} d \tau,
\end{equation}
which is not normalizable.
Since the orbit is invariant under the action of the group $\(\wh{U}_{C+R}(t,0)\)_{t\in\R}$, so is $\kett{\Psi}$, and $\wh{U}_{C+R}(t,0)\kett{\Psi}=\kett{\Psi}$ for any $t\in\R$.
Therefore,
\begin{equation}
\label{eq:constraint}
\wh{H}_{C+R}\kett{\Psi}=0.
\end{equation}

This connects Theorem \ref{thm:clock} and \emph{Canonical Quantum Gravity} \cite{Dewitt1967QuantumTheoryOfGravityI_TheCanonicalTheory}, where the Wheeler-DeWitt \emph{constraint equation} has the form \eqref{eq:constraint}. Its solutions are ``timeless'', but Page and Wootters \cite{PageWootters1983EvolutionWithoutEvolution} proposed that the dynamics is encoded in $\kett{\Psi}$ as correlations between system $R$ and the clock $C$, as in equation \eqref{eq:timeless}.
As we can see from Theorem \ref{thm:clock}, not only the clock's Hamiltonian, but the total Hamiltonian $\wh{H}_{C+R}$ as well has the form from eq. \eqref{eq:time_translation_Hamiltonian} when applied to quantum states, even though for $\tau$-extended quantum states it has the form \eqref{eq:constraint}.
Equation \eqref{eq:constraint} for $\tau$-extended quantum states provides a more timeless picture, in addition to the Parmenidean representation from Example \ref{ex:parmenidean}.
\qed
\end{remark}

\section{Worlds containing a sterile massless fermion}
\label{s:weyl}

Consider a Dirac fermion $\varphi$ that is massless and sterile, \ie it does not interact at all.
Since its chiral components are mixed only by the mass term, which is absent, $\varphi(t)$ decouples into two independent Weyl spinor fields $\varphi_\pm(t):\R^3\to\C^2$, where $\R^3$ is the Euclidean space.
In the Weyl representation, the Hamiltonian operators for the chiral components are
\begin{equation}
\label{eq:weyl_hamiltonian}
\wh{H}_\pm=\pm c\hbar\left(
-\sigma^x\frac{\partial\ }{\partial x}
-\sigma^y\frac{\partial\ }{\partial y}
-\sigma^z\frac{\partial\ }{\partial z}\right),
\end{equation}
where $\sigma^x$, $\sigma^y$, $\sigma^z$ are the Pauli matrices.

Let $\varphi=\varphi_+$ (or $\varphi=\varphi_-$) be a chiral planar wave. The basis of the Euclidean space $\R^3$ can be chosen so that $\varphi$ is independent of $x$ and $y$. 
Since the space coordinates decouple from the spin degrees of freedom, for $\varphi$,
\begin{equation}
\label{eq:weyl_dir}
\wh{H}_+\varphi=-c\sigma^z \otimes\hbar\frac{\partial}{\partial z}\varphi.
\end{equation}

This is a particular case of eq. \eqref{eq:interaction_hamiltonian_xp}, so the system is a {\TQS}.
If the planar wave $\varphi(t)$ is an eigenstate of $\wh{z}$ at $t=0$, it remains an eigenstate of $\wh{z}$ at all times $t$.

From Theorem \ref{thm:clock}, any world containing such a sterile massless fermion is a {\TQS}.

\section{Quantum representation of dynamical systems}
\label{s:dyn_sys}

In this Section I show that the quantum representation of a deterministic time-reversible dynamical system without time loops has the Hamiltonian \eqref{eq:time_translation_Hamiltonian}.

A \emph{deterministic time-reversible dynamical system} \cite{BrinStuck2002IntroductionToDynamicalSystems_ABSTRACT} $(\mc{S},\alpha)$ is an action $\alpha:\R\times\mc{S}\to\mc{S}$ of the additive group $(\R,+)$ on a set $\mc{S}$.
The set $\mc{S}$, containing the \emph{states} of the system, is endowed with a Lebesgue measure $\mu(A)=\int_A s\de s$ for any measurable subset $A$ of $\mc{S}$.
The action $\alpha$ is \emph{measure-preserving}, and represents the \emph{evolution law}. For any $s\in\mc{S}$, the orbit $\{\alpha(t,s)|t\in\R\}$ of the action $\alpha$ is a \emph{history} of the system, with the initial condition fixed by $s$. A periodic orbit will be called a \emph{time loop}.

The \emph{quantum representation} of a classical Hamiltonian system was introduced, in the Heisenberg picture, by Koopman and developed by von Neumann \cite{Koopman1931HamiltonianSystemsAndTransformationInHilbertSpace,vonNeumann1932KoopmanMethod}.
Here we will use the {\schrod} picture introduced by Sch\"onberg \cite{Schonberg1952TheActualKoopmanRepApplicationOf2ndQuantizationToClassicalI,Schonberg1953TheActualKoopmanRepApplicationOf2ndQuantizationToClassicalII}, generalized to a dynamical system $(\mc{S},\alpha)$.

We associate to $(\mc{S},\alpha)$ a quantum system on the Hilbert space $\hilbert_{\mc{S}}:=L^2(\mc{S},\C)$ of square-integrable complex functions on $\mc{S}$, in the following way.

\begin{enumerate}
	\item 
We define the scalar product of $\psi_1,\psi_2\in L^2(\mc{S},\C)$ by
\begin{equation}
\label{eq:dyn_sys_scalar_prod}
\braket{\psi_1}{\psi_2}:=\int_{\mc{S}}\psi_1^\ast(s)\psi_2(s)\de s.
\end{equation}

It satisfies $\braket{\psi_1}{\psi_2}=\braket{\psi_2}{\psi_1}^\ast$ for all $\psi_1,\psi_2\in L^2(\mc{S},\C)$, as required.
	\item 
The Dirac distributions on $\mc{S}$ form a canonical basis
\begin{equation}
\label{eq:dyn_sys_rep_basis}
\left(\ket{s}\right)_{s\in\mc{S}}.
\end{equation}

We define the one-parameter unitary group $\(\wh{U}_{(\mc{S},\alpha)}(t)\)_{t\in\R}$ acting as evolution operators on the basis vectors from \eqref{eq:dyn_sys_rep_basis}, so that for any $t\in\R$ and $s\in\mc{S}$,
\begin{equation}
\label{eq:dyn_sys_rep_evol}
\wh{U}_{(\mc{S},\alpha)}(t)\ket{s}=\ket{\alpha(t,s)}.
\end{equation}
\end{enumerate}

Since for any $\psi\in L^2(\mc{S},\C)$
\begin{equation}
\label{eq:dyn_sys_wavefunction}
\ket{\psi}=\int_{\mc{S}}\psi(s)\ket{s}\de s,
\end{equation}
it follows by linearity from equation \eqref{eq:dyn_sys_rep_evol} that
\begin{equation}
\label{eq:dyn_sys_Koopman_op}
\wh{U}_{(\mc{S},\alpha)}(t)\ket{\psi}:=\int_{\mc{S}}\psi(s)\ket{\alpha(t,s)}\de s.
\end{equation}

The one-parameter unitary group \eqref{eq:dyn_sys_Koopman_op} is generated by a Hermitian operator $\wh{H}_{(\mc{S},\alpha)}$, $\wh{U}_{(\mc{S},\alpha)}(t):=e^{-\frac{i}{\hbar}t\wh{H}_{(\mc{S},\alpha)}}$.
The histories $\{\ket{\alpha(s,t)}|t\in\R\}$ satisfy the {\schrod} equation \eqref{eq:schrod-TQS} with the initial condition $\ket{\alpha(s,0)}=\ket{\alpha(s)}$.

\begin{example}
\label{ex:dyn_sys_classical_Hamiltonian_sys}
The representation was developed starting with classical Hamiltonian systems \cite{Schonberg1952TheActualKoopmanRepApplicationOf2ndQuantizationToClassicalI,Schonberg1953TheActualKoopmanRepApplicationOf2ndQuantizationToClassicalII,Koopman1931HamiltonianSystemsAndTransformationInHilbertSpace}.
The state space $\mc{S}$ is the \emph{phase space}, assumed here to be $\mc{S}=\R^{2N}$, and having canonical coordinates $(\p,\q):=(p_1,\ldots,p_N,q^1,\ldots,q^N)$.
The phase space $\mc{S}$ is a symplectic manifold, and the measure is preserved automatically by the evolution law due to Liouville's theorem.
The action $\alpha:\R\times\mc{S}\to\mc{S}$ associates to any initial state $s(0)$ the state $s(t)$ resulting from classical dynamics.
In terms of the classical Hamiltonian function $\mc{H}$, the quantum Hamiltonian is the Liouville operator,
\begin{equation}
\label{eq:hamiltonian_sys}
\wh{H}_{(\mc{S},\alpha)}=-i\hbar\sum_{j=1}^N\left(
\frac{\partial\mc{H}}{\partial p_j} \frac{\partial}{\partial q^j}
- \frac{\partial\mc{H}}{\partial q^j} \frac{\partial}{\partial p_j} 
\right).
\end{equation}
Then a classical probability density $\rho(\p,\q)$ on $\mc{S}$ satisfies the classical Liouville equation.
This extends by linearity to a complex function $\psi(\p,\q)$ on $\mc{S}$, and therefore $\ket{\psi}$ satisfies the {\schrod} equation for the Hamiltonian $\wh{H}_{(\mc{S},\alpha)}$.
\end{example}

\begin{theorem}
\label{thm:quantum_representation}
Let $(\mc{S},\alpha)$ be a deterministic time-reversible dynamical system which is not ergodic and doesn't have time loops.
Then, the Hamiltonian operator is $\wh{H}_{(\mc{S},\alpha)}=-i\hbar\frac{\partial\ }{\partial\tau}$, and at all times the states are eigenstates of its canonical conjugate $\wh{\tau}$.
\end{theorem}
\begin{proof}
The representation \eqref{eq:dyn_sys_Koopman_op} associates to each non-periodic orbit of the action $\alpha$ a history $\(\wh{U}_{(\mc{S},\alpha)}(t)\ket{s_0}\)_{t\in\R}$ consisting of mutually orthogonal states. Then,
the one-parameter unitary group $\(\wh{U}_{(\mc{S},\alpha)}(t)\)_{t\in\R}$ acts like a translation group on the mutually orthogonal vectors $\{\wh{U}_{(\mc{S},\alpha)}(t)\ket{s_0}|t\in\R\}$.
Therefore, the restriction of $\wh{H}_{(\mc{S},\alpha)}$ to the subspace spanned by $\(\wh{U}_{(\mc{S},\alpha)}(t)\ket{s_0}\)_{t\in\R}$ can be put in the form $-i\hbar\frac{\partial\ }{\partial\tau}$.
At all times, the states of the system are eigenstates of the canonical conjugate $\wh{\tau}$ of $-i\hbar\frac{\partial\ }{\partial\tau}$.
Since in the absence of time loops $\mc{S}$ is partitioned into non-periodic orbits, the result extends to the entire Hilbert space $\hilbert$.
\end{proof}

An observable of the original dynamical system is a function $f:\mc{S}\to\R$. Any observable can be uniquely represented by a Hermitian operator
\begin{equation}
\label{eq:dyn_sys_observable}
\wh{f}=\int_\mc{S} f(s)\ket{s}\bra{s}\de s.
\end{equation}

\begin{remark}
\label{rem:commutativity}
If the original dynamical system is a classical Hamiltonian system, the state space $\mc{S}$ is a classical phase space.
The regular multiplication of the phase space coordinates $(p_a,q^a)\in\mc{S}$ commutes, $p_a q^a=q^a p_a$.
Therefore, the corresponding operators in the quantum representation also commute, $[\wh{q}^a,\wh{p}_a]=0$. This is because the operator $\wh{p}_a$ \emph{is not} the canonical conjugate of $\wh{q}^a$! In fact, any two observables $f_1,f_2:\mc{S}\to\R$ commute, being just real functions, and the operators representing them also commute. The operators $\wh{q}^a$ and its canonical conjugate $-i\hbar\frac{\partial\ }{\partial q^a}$ satisfy the canonical commutation relations $[\wh{q}^a,-i\hbar\frac{\partial\ }{\partial q^a}]=i\hbar\wh{I}$, but $\wh{p}_a$, the classical momentum $p_a$, is different from the canonical conjugate $-i\hbar\frac{\partial\ }{\partial q^a}$, which does not represent any classical observable.
When applied to a state vector $\ket{s}$, representing the classical state $s$, the operators $-i\hbar\frac{\partial\ }{\partial q^a}$ result in superpositions of state vectors representing classical states.
The operators $\wh{f}$ representing classical observables $f:\mc{S}\to\R$, in particular the classical momenta $\wh{p}_a$, only multiply $\ket{s}$ with the value of the classical observable $f(s)$ of the classical state $s$.
\end{remark}

\section{Diversity from the physical meaning of the observables}
\label{s:HSF}

Before discussing Objection \ref{obj:simple}, let us understand it exactly.

Let us consider two {\TQS}s specified by the triples,
$(\hilbert_1,\wh{H}_1,\ket{\psi_1(t)})$ and $(\hilbert_2,\wh{H}_2,\ket{\psi_2(t)})$, with the operators $\wh{\tau}_1$ and respectively $\wh{\tau}_2$.
Due to the $\tau$-translation symmetry, all eigenspaces of the operator $\wh{\tau}$ for a {\TQS} have the same dimension, denoted here by $d(\wh{H})$, equal to the multiplicity of the energy eigenspaces. 
If $d(\wh{H}_1)=d(\wh{H}_2)$, there is a unitary map $\wh{M}:\hilbert_1\to\hilbert_2$, so that $\wh{H}_2=\wh{M} \wh{H}_1 \wh{M}^{-1}$.

\begin{proposition}
\label{thm:TQS-ambiguity}
For any fixed $t_1,t_2\in\R$, we can choose $\wh{M}$ so that $\wh{M}\ket{\psi_1(t_1)}=\ket{\psi_2(t_2)}$.
\end{proposition}
\begin{proof}
Since $\wh{M}\ket{\psi_1(t_1)},\ket{\psi_2(t_2)}\in\hilbert_2$ are eigenvectors of $\wh{\tau}_2$, let $\tau_1$ and $\tau_2$ be their eigenvalues.
Then, the vector $\wh{M}\wh{U}_1(\tau_2-\tau_1)\ket{\psi_1(t_1)}=\wh{U}_1(\tau_2-\tau_1)\wh{M}\ket{\psi_1(t_1)}\in\hilbert_2$ corresponds to the eigenvalue $\tau_2$.
The representation from Example \ref{ex:parmenidean} ensures the existence of a unitary operator $\wh{S}_2$ on $\hilbert_2$ that commutes with $\wh{H}_2$, so that $\wh{S}_2\wh{M}\wh{U}_1(\tau_2-\tau_1)\ket{\psi_1(t_1)}=\ket{\psi_2(t_2)}$. Then, all we have to do is to replace $\wh{M}$ with $\wh{S}_2\wh{M}\wh{U}_1(\tau_2-\tau_1)$.
\end{proof}

Then, since $\wh{M}\ket{\psi_1(t_1+t)}=\ket{\psi_2(t_2+t)}$ for any $t\in\R$, even the histories of the two {\TQS}s $(\hilbert_1,\wh{H}_1,\ket{\psi_1(t)})$ and $(\hilbert_2,\wh{H}_2,\ket{\psi_2(t)})$ are unitarily equivalent.

This suggests a more precise restatement of Objection \ref{obj:simple}:

\begin{objectionbis}{obj:simple}
\label{obj:HSF}
Two quantum systems specified by the triples $(\hilbert_1,\wh{H}_1,\ket{\psi_1(t)})$ and $(\hilbert_2,\wh{H}_2,\ket{\psi_2(t)})$ with the same multiplicity of the eigenspaces of $\wh{H}$ are the same.
Therefore, there is no difference between the various examples from this article (up to the multiplicity). There is basically only one such system, and it is too simple to be interesting.
\end{objectionbis}

Objection \ref{obj:HSF} is based on an assumption called the \emph{Hilbert Space Fundamentalism Thesis} (\HSF), as coined by Carroll \cite{Carroll2021RealityAsAVectorInHilbertSpace}, although this thesis was assumed under various other names \cite{Piazza2010GlimmersOfAPreGeometricPerspective,Tegmark2015ConsciousnessAsAStateOfMatter,CarrollSingh2019MadDogEverettianism,Aaronson2021TheZenAntiInterpretationOfQuantumMechanics}:
\setthesisCustomtag{HSF}
\begin{thesisCustom}
\label{thesis:HSF}
Everything about a physical system, including the $3$D-space, a preferred basis, a preferred factorization of the Hilbert space (a tensor product structure, needed to represent subsystems, {\eg} particles), emerge uniquely from the triple
\begin{equation}
	\label{eq:MQS}
	(\hilbert,\wh{H},\ket{\psi(t)}).
\end{equation}
\end{thesisCustom}

No other information is assumed to be needed, for example a preferred basis or additional observables to which we assign a particular physical meaning like position operators.
More precisely, Thesis {\HSF} states that the following are sufficient to give a complete description of a quantum system:
\begin{enumerate}
	\item 
\textbf{The dimension of the Hilbert space $\hilbert$}: \emph{``a particular choice of Hilbert space is completely specified by its dimension''} (\cite{Carroll2021RealityAsAVectorInHilbertSpace}, p. 4).
Any two Hilbert spaces of the same dimension are unitarily equivalent.
Any two infinite-dimensional separable Hilbert spaces are unitarily equivalent.
	\item 
\textbf{The spectrum of the Hamiltonian} (including the multiplicities of the eigenvalues): \emph{``the laws of physics are determined solely by the energy eigenspectrum of the Hamiltonian''} (\cite{Carroll2021RealityAsAVectorInHilbertSpace}, p. 1). This is because the Hamiltonian can be diagonalized by a unitary transformation, and there is no preferred basis.
	\item 
\textbf{The state \emph{vector} $\ket{\psi(t)}$}, taken as abstract, basis-independent vector, and not as a wavefunction $\braket{\q}{\psi(t)}$, which would assume again a preferred basis $\(\ket{\q}\)_{\q\in\mc{C}}$: \emph{``the fundamental ontology of the world is completely and exactly represented by a vector in an
abstract Hilbert space, evolving in time according to unitary {\schrod} dynamics''} (\cite{Carroll2021RealityAsAVectorInHilbertSpace}, p. 2).
\end{enumerate}

Sometimes Quantum Mechanics is introduced based only on the triple $(\hilbert,\wh{H},\ket{\psi(t)})$ from \eqref{eq:MQS}.
But when we study physical examples the observables, represented by self-adjoint operators on $\hilbert$, have attached clear physical meanings.
For example, consider the set of position observables $\wh{q}_j$, $j\in\{1,\ldots,3n\}$, where $n$ is the number of particles in a three-dimensional space.
If a transformation $\wh{S}$ of $\hilbert$ that commutes with $\wh{H}$ is applied to these observables, we obtain different observables $\wh{S}\wh{q}_j\wh{S}^{-1}$. If the Hamiltonian has a particular form in terms of the observables $\wh{q}_j$, in the new basis it will have the same form in terms of the observables $\wh{S}\wh{q}_j\wh{S}^{-1}$. But in standard Quantum Mechanics we never identify the observables $\wh{S}\wh{q}_j\wh{S}^{-1}$ with positions (unless, of course, $\wh{S}$ represents a combination between an isometry of space, a permutation of particles of the same type, and a gauge transformation).
This means that in standard Quantum Mechanics, in addition to the triple $(\hilbert,\wh{H},\ket{\psi(t)})$, we also need to identify what each observable means physically, \ie whether it is a position or a momentum observable or what other physical property it represents.
Without this, all states would be identical unit vectors.
But the association of physical meanings to observables allows us to represent a state vector $\ket{\psi}$ as a wavefunction $\psi(q_1,q_2,\ldots):=\braket{q_1,q_2,\ldots}{\psi}$.
In addition, for more than one particle, the position observables are positions of configurations of particles, so $\psi(q_1,q_2,\ldots)$ is a wavefunction in the position configuration space $\mc{C}$. To recover the $3$D-space $\R^3$, in the case of a fixed number of particles, we need a way to decompose the configuration space as a Cartesian product of $n$ copies of $\R^3$, $\mc{C}=\R^3\times\R^3\times\ldots$, and an identification between their points. This decomposition corresponds to a \emph{tensor product structure} of the total Hilbert space, $\hilbert\cong L^2(\R^3)\otimes L^2(\R^3)\otimes\ldots$.
More complex structures are built from collections of observables, for example how space is represented as a triple of possible eigenvalues of position operators.
We see that even in Non-Relativistic Quantum Mechanics we need a much richer structure than an abstract unit vector and the Hamiltonian's spectrum.

Therefore, a ``less Zen'' \cite{Aaronson2021TheZenAntiInterpretationOfQuantumMechanics} definition of a quantum theory requires more than just the triple $(\hilbert,\wh{H},\ket{\psi(t)})$:
\begin{definition}[Quantum theory]
\label{def:quantum-theory}
A \emph{quantum theory} is given by specifying
\begin{enumerate}
	\item
	A triple $(\hilbert,\wh{H},\ket{\psi(t)})$ as in equation \eqref{eq:MQS}.
	\item
	A set of self-adjoint operators $\{\wh{A}_1,\wh{A}_2,\ldots\}$ on $\hilbert$.
	\item
	Optionally, a tensor product structure $\hilbert\cong\hilbert_1\otimes\hilbert_2\otimes\ldots$, so that some of the operators can be expressed as a product of operators on the factors $\hilbert_k$ of the form
\begin{equation}
\label{eq:TPS-obs}
	\wh{A}_j=\ldots\otimes\underbrace{\wh{I}}_{k-1}\otimes\underbrace{\wh{\wt{A}}_j}_{k}\otimes\underbrace{\wh{I}}_{k+1}\otimes\ldots.
\end{equation}
That is, the observables corresponding to subsystems are also observables of the total system.
The tensor product structure can be more complicated in the case of identical particles, its definition involving also direct sums of spaces and symmetry/anti-symmetry under permutations.

The tensor product structure is optional because, as explained in \cite{ZanardiLidarLloyd2004QuantumTensorProductStructuresAreObservableInduced}, the tensor product itself can be encoded in a set of algebras of observables satisfying certain locality and completeness conditions that allow them to be interpreted as acting locally on the factors of the Hilbert space. 
	\item
	An expression of the Hamiltonian operator in terms of operators from $\{\wh{A}_1,\wh{A}_2,\ldots\}$ (for example in terms of the position and momentum operators).
	\item
	A \emph{dictionary} that associates to each operator $\wh{A}_j$ a physical meaning, \ie what physical property it represents.
	For example, the dictionary can specify what operators represent position or momentum observables, and for which particle.
	\item
Other structures, like the $3$D-space structure, can be specified too, although they can be derived from the position observables and the tensor product structure, at least in Non-Relativistic Quantum Mechanics.
\end{enumerate}
\end{definition}

And often Quantum Mechanics is presented in a way that clearly introduces the tensor product structure, and at least the position and momentum observables and other observables depending on them (like angular momentum), and the spin observables.
And sometimes it is presented only as a triple $(\hilbert,\wh{H},\ket{\psi(t)})$ as in equation \eqref{eq:MQS}.

But can we get rid of the dictionary that associates physical properties to the operators, so that all we need is the triple $(\hilbert,\wh{H},\ket{\psi(t)})$, from which they are supposed to emerge?
Thesis {\HSF} claims that all we need is the triple $(\hilbert,\wh{H},\ket{\psi(t)})$, and a recipe to recover uniquely the other structures, and their physical meaning becomes clear from these. That no tensor product structure, no position observables, no position basis, no position configuration space, and nothing of this sort based on other observables and their physical meaning need to be provided \emph{a priori}, because they can, supposedly, ``emerge'' uniquely from $(\hilbert,\wh{H},\ket{\psi(t)})$.

Thesis {\HSF} is assumed for example in \cite{Carroll2021RealityAsAVectorInHilbertSpace,Aaronson2021TheZenAntiInterpretationOfQuantumMechanics}, but also in attempts to show that preferred pointer states or a preferred basis can emerge due to decoherence \cite{Zurek1998DecoherenceEinselectionAndTheExistentialInterpretation}.
To avoid restraining the possible solutions to the problem of the preferred pointer observable, no \emph{a priori} physical meaning is assigned to the observables, leaving open the freedom to decide, after decoherence selects a pointer basis, what physical observables it represents.

Some authors interested in decoherence and the Many-Worlds Interpretation (MWI) believe that preferred pointer observables are selected by decoherence, and after that we can interpret them physically. This is why Thesis {\HSF} can be encountered among supporters of the MWI like \cite{Piazza2010GlimmersOfAPreGeometricPerspective,Tegmark2015ConsciousnessAsAStateOfMatter,CarrollSingh2019MadDogEverettianism,Aaronson2021TheZenAntiInterpretationOfQuantumMechanics}.
But not all of its supporters accept the {\HSF}. For example, Lev Vaidman is clear about the necessity of the position observables and the $3$D-space, acknowledging that the physical content is given by the wavefunction, and not merely in the state vector as an abstract unit vector \cite{Vaidman2016AllIsPsi,SEP-Vaidman2021MWI}.

Thesis {\HSF} is not exclusive to the MWI.
A commitment to Thesis {\HSF} may be observed sometimes also in the Quantum Information community, which mainly works in abstract Hilbert spaces. The operators and the state vectors are often thought of as disembodied mathematical entities, and the correspondence with physical reality is seen merely as an implementation. Since information is considered substrate-independent, whether we represent a qubit as a spin $1/2$ particle or any other two-level system becomes irrelevant, from the standpoint of the information.
In Quantum Gravity we can also see a commitment to the {\HSF}, for example by considering the ``clock ambiguity'' \cite{Albrecht1995TheoryOfEeverythingVsTheoryOfAnything,AlbrechtIglesias2008ClockAmbiguityAndTheEmergenceOfPhysicalLaws,AlbrechtIglesias2012ClockAmbiguityImplicationsNewDevelopments} to be a problem for the Page-Wootters formalism, and attempts to resolve it assuming {\HSF} as in \cite{MarlettoVedral2017EvolutionWithoutEvolutionAndWithoutAmbiguities}.
I will return to this in Section \sref{s:time-operator-clock}.

In \cite{Piazza2010GlimmersOfAPreGeometricPerspective,Tegmark2015ConsciousnessAsAStateOfMatter,CarrollSingh2019MadDogEverettianism,Carroll2021RealityAsAVectorInHilbertSpace,Aaronson2021TheZenAntiInterpretationOfQuantumMechanics}, Thesis {\HSF} is inherently associated with the MWI, perhaps also because of the problem of the preferred pointer observable.
However, as we have seen, {\HSF} is not exclusive to the MWI community.
Moreover, Carroll \& Singh's \cite{CarrollSingh2019MadDogEverettianism} attempt to reconstruct the space from the Hamiltonian and the unit vector alone, does not make use of anything specific to the MWI.
Their construction has two major steps:
\begin{enumerate}
	\item 
	Recovering the topology of space from the Hamiltonian's spectrum, based on \cite{CotlerEtAl2019LocalityFromSpectrum}.
	To this end, they consider that bounded regions of space have associated finite-dimensional Hilbert spaces, and unions of disjoint regions have associated the tensor product of their Hilbert spaces.
	In \cite{CotlerEtAl2019LocalityFromSpectrum}, it was claimed that from the Hamiltonian's spectrum and a condition involving the locality of the Hamiltonian it is possible to recover a unique tensor product structure of the Hilbert space. Carroll \& Singh use this claim (refuted mathematically in \cite{Stoica2022SpaceThePreferredBasisCannotUniquelyEmergeFromTheQuantumStructure} and \cite{Stoica2024DoesTheHamiltonianDetermineTheTPSAndThe3dSpace})) to interpret the factors in the tensor product as points or small regions in space.
	\item 
	Recovering the distances between points, based on the mutual information, which requires the state vector $\ket{\psi}$ as an input.
\end{enumerate}

As we can see, none of these steps is specific to MWI. In fact, MWI may jeopardize this construction, because different worlds may have different metric tensors, due to different distributions of the masses.
But their construction of space doesn't seem to take this into account, so it is difficult to see how it can be compatible simultaneously with both MWI and Quantum Gravity.

Only in the next step in \cite{CarrollSingh2019MadDogEverettianism}, the emergence of classicality, which corresponds to the emergence of the pointer states through decoherence, MWI is invoked.

Giddings proposed another reconstruction motivated by Quantum Gravity \cite{Giddings2019QuantumFirstGravity}.

We can identify two major types of ambiguities that can contradict Thesis {\HSF}, based on how it fails to uniquely recover a quantum theory (as in Definition \ref{def:quantum-theory}):

\begin{type}[Ambiguity between distinct theories]
\label{type:ambiguity-inter-theoretical}
The first type, already known, is the existence of \emph{dual theories}, two non-equivalent quantum theories (as in Definition \ref{def:quantum-theory}), having different types of observables and different physical meanings associated with them, represented by the same triple $(\hilbert,\wh{H},\ket{\psi(t)})$.

A simple situation is Schmelzer's construction, using the KdV equation, of examples of Hamiltonians admitting different sets of position and momentum observables, and different tensor product structures \cite{Schmelzer2009WhyTheHamiltonOperatorAloneIsNotEnough}.

A previously known and more famous example is the AdS/CFT correspondence \cite{Maldacena1999AdSCFTInitial,Kaplan2016LecturesOnAdSCFTFromTheBottomUp}.
It was mentioned in \cite{Stoica2022SpaceThePreferredBasisCannotUniquelyEmergeFromTheQuantumStructure}, \S{VI}, Problem $4$ as a different type of ambiguity than the main type considered there (Type \ref{type:ambiguity-intra-theoretical}).
In the AdS/CFT correspondence, even if the Hilbert space of the AdS theory of Quantum Gravity in the bulk is unitarily equivalent to the Hilbert space of the conformal field theory on the boundary, these theories have different physical meanings. For example their spaces have different numbers of dimensions, because the latter is the boundary of the former. Observables from the AdS theory translate into observables with different physical meaning from the CFT theory.

The existence of dualities like AdS/CFT is known, but one may think that they are exceptional situations. The name ``duality'' itself suggests that three or more distinct quantum theories represented on the same Hilbert space are probably even more exceptional.
But later in this Section we will see that there can be infinitely many quantum theories with unitarily equivalent triples $(\hilbert,\wh{H},\ket{\psi(t)})$.
\end{type}

However, the attempts to reconstruct structures like the tensor product structure and the $3$D-space are based on recipes that specify certain requirements expected for the resulting structures.
An example of such a requirement is the locality of the interaction in the underlying physical space of the theory. If the CFT theory in dimension $d$ is local, then it seems very difficult for the dual AdS theory, in a $d+1$-dimensional space, to also be local, and if they can both be local, the two characterizations of the locality are different, so local states in the bulk are not the same as local states on the boundary \cite{AlmheiriDongHarlow2015BulkLocalityAndQuantumErrorCorrectionInAdSCFT}.
It seems therefore justified to think that, by stronger specifications of the recipe to obtain a quantum theory as in Definition \ref{def:quantum-theory}, for example regarding the number of dimensions, we can exclude dualities and recover a unique quantum theory from the triple $(\hilbert,\wh{H},\ket{\psi(t)})$ alone.
This leads to another type of ambiguity:

\begin{type}[Ambiguity of recovering the same theory]
\label{type:ambiguity-intra-theoretical}
The second type of ambiguity is the existence of physically different ways to read a particular quantum theory from a given triple $(\hilbert,\wh{H},\ket{\psi(t)})$.
Even if we have a recipe to identify the expected quantum theory (as in Definition \ref{def:quantum-theory}, \ie observables and their correspondence with the physical properties from the triple $(\hilbert,\wh{H},\ket{\psi(t)})$), as long as this recipe is invariant under the transformations that preserve $(\hilbert,\wh{H},\ket{\psi(t)})$, the same recipe will give physically different results \cite{Stoica2022SpaceThePreferredBasisCannotUniquelyEmergeFromTheQuantumStructure}.
For example, a unitary transformation that commutes with the Hamiltonian can lead to different tensor product structures and position observables, resulting in non-equivalent $3$D-space structures. On these, the same unit vector $\ket{\psi(t)}$ can appear as distinct wavefunctions, representing physically distinct worlds.
\end{type}

But any recipe for the emergence of the $3$D-space or of other structures, or for the emergence of preferred pointer states, leads to Type \ref{type:ambiguity-intra-theoretical} ambiguities, as shown in \cite{Stoica2022SpaceThePreferredBasisCannotUniquelyEmergeFromTheQuantumStructure}. In fact, the same triple $(\hilbert,\wh{H},\ket{\psi(t)})$ can represent completely different physical worlds, as different as the world differs at different moments of time \cite{Stoica2022SpaceThePreferredBasisCannotUniquelyEmergeFromTheQuantumStructure}, or as different even as the worlds corresponding to different measurement outcomes can be \cite{Stoica2023PrinceAndPauperQuantumParadoxHilbertSpaceFundamentalism}.
Moreover, the claim from \cite{CotlerEtAl2019LocalityFromSpectrum}, that under appropriate conditions the tensor product structure emerges uniquely from the spectrum of the Hamiltonian, is contradicted both by the results from \cite{Stoica2022SpaceThePreferredBasisCannotUniquelyEmergeFromTheQuantumStructure}, and especially from \cite{Stoica2024DoesTheHamiltonianDetermineTheTPSAndThe3dSpace}, where it is shown that the possible physically distinct tensor product structures can be parametrized by a number of parameters that increases exponentially with the number of factors.
A proof that there are infinitely many physically distinct worlds described by the same triple $(\hilbert,\wh{H},\ket{\psi(t)})$ and the same \emph{fixed} tensor product structure is given in \cite{Stoica2023PrinceAndPauperQuantumParadoxHilbertSpaceFundamentalism}.
Therefore, there is nothing remotely unique in the structures obtained by constructions based on {\HSF}.

The abundance of physically distinct solutions represented by the same triple $(\hilbert,\wh{H},\ket{\psi(t)})$, shown in \cite{Stoica2022SpaceThePreferredBasisCannotUniquelyEmergeFromTheQuantumStructure,Stoica2023PrinceAndPauperQuantumParadoxHilbertSpaceFundamentalism,Stoica2024DoesTheHamiltonianDetermineTheTPSAndThe3dSpace}, is of Type \ref{type:ambiguity-intra-theoretical}, \ie it is about theories that have the same physical meanings associated with observables, for example the same number of space dimensions.
The difference between them is due to the fact that this association is different. Because of this, we can identify for example, in the same triple $(\hilbert,\wh{H},\ket{\psi(t)})$, different position configuration spaces and different $3$D-space structures. This implies that the same $\ket{\psi(t)}$ can correspond to different wavefunctions, and therefore it can represent different physical states.
The different ways to associate physical meaning to the operators can be related by unitary transformations that commute with the Hamiltonian, and these are very abundant.

\begin{replyToObjection}{obj:simple}
\label{reply:obj:HSF}
From the statement of Thesis {\HSF} it follows that to contradict it, it is sufficient to find two {\TQS}s that
\begin{enumerate}
	\item[1)] 
are unitarily equivalent
	\item[2)] 
but represent different physical systems.
\end{enumerate}

The following Corollary shows that the examples from this article provide infinitely many concrete counterexamples to the {\HSF} Thesis, because the physical content of the operators representing observables in the two systems can be non-equivalent.

\begin{corollary}
\label{thm:HSF}
There are infinitely many quantum systems with the same triple $(\hilbert,\wh{H},\ket{\psi(t)})$ but completely different physical content, exhibiting ambiguities of both Type \ref{type:ambiguity-inter-theoretical} and Type \ref{type:ambiguity-intra-theoretical}.
\end{corollary}
\begin{proof}
Let $(\hilbert_1,\wh{H}_{\tau_1},\ket{\psi_1(t)})$ and $(\hilbert_2,\wh{H}_{\tau_2},\ket{\psi_2(t)})$ be two {\TQS}s. 
From Example \ref{ex:parmenidean} we recall that the Hilbert space $\hilbert$ of a {\TQS} has the form $\hilbert\cong\hilbert_{C}\otimes L^2\(\R,\C\)$.
Then, let $\hilbert_1\cong\hilbert'_1\otimes L^2\(\R,\C\)$ and $\hilbert_2\cong\hilbert'_2\otimes L^2\(\R,\C\)$.
If both the Hilbert spaces $\hilbert'_1$ and $\hilbert'_2$ are finite-dimensional and have the same dimension, they are unitarily equivalent. If both are infinite-dimensional and separable, they are unitarily equivalent. In both cases, there is a unitary map $\wh{M}:\hilbert_1\to\hilbert_2$ . Due to the Stone-von Neumann Theorem (see Remark \ref{rem:motivation} or {\eg} \cite{Hall2013QuantumTheoryForMathematicians}, p. 239), $\wh{M}$ can be chosen so that $\wh{H}_{\tau_2}=\wh{M} \wh{H}_{\tau_1} \wh{M}^{-1}$ and $\wh{\tau}_2=\wh{M} \wh{\tau}_1 \wh{M}^{-1}$.
Then $\wh{M}$ maps eigenvectors of $\wh{\tau}_1$ into eigenvectors of $\wh{\tau}_2$.
Moreover, since for any fixed $t_1,t_2\in\R$, $\ket{\psi_1(t_1)}$ is an eigenvector of $\wh{\tau_1}$ and $\ket{\psi_2(t_2)}$ is an eigenvector of $\wh{\tau_2}$, the translational symmetry in the $\tau$-spaces allows $\wh{M}$ to be chosen so that $\wh{M}\ket{\psi_1(t_1)}=\ket{\psi_2(t_2)}$. Since $\wh{H}_{\tau_2}=\wh{M} \wh{H}_{\tau_1} \wh{M}^{-1}$, $\wh{M}\ket{\psi_1(t_1+t)}=\ket{\psi_2(t_2+t)}$ for any $t\in\R$, so even their histories are equivalent.

Consider any two of the examples of {\TQS}s from Sections \sref{s:measurements}--\sref{s:dyn_sys} for which $\hilbert'_1$ and $\hilbert'_2$ as above have the same dimension. Then, the two {\TQS}s are unitarily equivalent.

If the two examples of unitarily equivalent {\TQS}s have different types of observables, space dimension, or dynamics, this illustrates the existence of ambiguities of Type \ref{type:ambiguity-inter-theoretical}, or dualities, except that there are infinitely many of them, all sharing the same triple $(\hilbert_1,\wh{H}_{\tau},\ket{\psi(t)})$.

Moreover, since we can freely choose $t_1,t_2\in\R$, the first system is equivalent at any time $t_1$ with the second system at any time $t_2$, and because unitary equivalence is transitive, the first system at a given time is equivalent to itself at any other time. Therefore, the structures from Thesis {\HSF} are unable even to capture the physical differences of the same system at different times. Also since we can freely choose $\ket{\psi_1(t_1)}$ and $\ket{\psi_2(t_2)}$, as long as they are eigenvectors of $\wh{\tau_1}$, respectively of $\wh{\tau_2}$, any history of the first system is equivalent to any history of the second system, and in particular to any history of itself. Therefore, the structures from Thesis {\HSF} are unable to distinguish physical differences between two possible alternative histories of the same system, and even between two possible alternative histories of systems governed by different physical laws, as long as they are translational.
This illustrates the existence of ambiguities of Type \ref{type:ambiguity-intra-theoretical}, and of mixed ambiguities of both Type \ref{type:ambiguity-inter-theoretical} and Type \ref{type:ambiguity-intra-theoretical}.

Therefore, Thesis {\HSF} is contradicted by the existence of ambiguities between infinitely many physically distinct systems from the same or from different quantum theories (as in Definition \ref{def:quantum-theory}).

For a particular example, recall Theorem \ref{thm:quantum_representation}.
A Hamiltonian function $\mc{H}(p_a,q^a)$ is a function $\mc{H}:\mc{S}\to\R$ of the phase space coordinates $(p_a,q^a)\in\mc{S}$.
Consider two classical Hamiltonian systems with Hamiltonian functions $\mc{H}_1(p_{1 a},q_1^a)$ and $\mc{H}_2(p_{2 a},q_2^a)$. 
We assume that the two classical systems consist of point particles, and the configuration spaces are Euclidean spaces of the same dimension.
Under the conditions of Theorem \ref{thm:quantum_representation}, their quantum representations are both {\TQS}s.
Let the {\TQS}s $(\hilbert_1,\wh{H}_{\tau_1},\ket{\psi_1(t)})$ and $(\hilbert_2,\wh{H}_{\tau_2},\ket{\psi_2(t)})$ be their quantum representations. 
Then, there is a unitary map $\wh{M}:\hilbert_1\to\hilbert_2$, so that $\wh{H}_{\tau_2}=\wh{M} \wh{H}_{\tau_1} \wh{M}^{-1}$ and $\wh{M}\ket{\psi_1(t_1)}=\ket{\psi_2(t_2)}$ for any fixed $t_1,t_2\in\R$. Then, because $\wh{H}_{\tau_2}=\wh{M} \wh{H}_{\tau_1} \wh{M}^{-1}$, $\wh{M}\ket{\psi_1(t_1+t)}=\ket{\psi_2(t_2+t)}$ for any $t\in\R$, so even their histories are equivalent.

While $\wh{M}$ maps $\wh{H}_{\tau_1}$ into $\wh{H}_{\tau_2}$, it does not necessarily map $(\wh{p}_{1 a},\wh{q}_1^a)$ into $(\wh{p}_{2 a},\wh{q}_2^a)$, because it is possible that $\mc{H}_1(p_{1 a},q_1^a)\neq \mc{H}_2(p_{2 a},q_2^a)$, even though both quantum representations have translational Hamiltonians $\wh{H}_{\tau_1}=-i\hbar\frac{\partial\ }{\partial\tau_1}$ and $\wh{H}_{\tau_2}=-i\hbar\frac{\partial\ }{\partial\tau_2}$.
In general, if the two classical dynamical systems are non-equivalent under canonical transformations, the values of the observables $(\wh{p}_{1 a},\wh{q}_1^a)$ and $(\wh{p}_{2 a},\wh{q}_2^a)$ evolve differently, and are not equivalent under the unitary map $\wh{M}$.
Even if the quantum representation of different dynamical systems have the same quantum Hamiltonian, the information from the original dynamical system is not lost, because it is preserved in the corresponding observables $\wh{f}$ defined in equation \eqref{eq:dyn_sys_observable}. But this information cannot be extracted from $(\hilbert_1,\wh{H}_{\tau_1},\ket{\psi_1(t)})$ and $(\hilbert_2,\wh{H}_{\tau_2},\ket{\psi_2(t)})$ alone, we also need the observables. 

In particular, the physical spaces of the two classical Hamiltonian systems can have different dimensions, and the systems factorize differently into subsystems. To see this, consider two systems of classical particles, with different space dimensions $d_1\neq d_2$. Then, we can construct a counterexample by choosing the number of particles so that the dimensions of the two phase spaces are equal.
We can choose for example the first system to contain $n d_2$ particles, and the second system to contain $n d_1$ particles, where $n$ is a positive integer. Then, the dimensions of the phase spaces of the two systems are equal. 
If the dimension of the physical space is different for the two systems, the factorization into subsystems is also different.
The classical dynamics is different for the two systems, because the particles have different degrees of freedom available.
For example, in three dimensions the Laplace equation leads to inverse-square forces, while in four dimensions it leads to inverse-cube forces.
In the quantum representation, the difference is reflected in the quantum observables that represent the classical positions and momenta.
If there was a procedure by which space dimension can be unambiguously extracted from the Hamiltonian and the state vector, it should result in the same space dimension, but $d_1\neq d_2$.
Therefore, there are differences in the number of space dimensions, but also in the dynamics, but these differences are lost if we ignore the additional observables like the positions and momenta.
Therefore, the triples $(\hilbert_1,\wh{H}_{\tau_1},\ket{\psi_1(t)})$ and $(\hilbert_2,\wh{H}_{\tau_2},\ket{\psi_2(t)})$ are not sufficient to capture the entire physical content.

This shows that the physical content of a quantum system is almost entirely missed by the triple $(\hilbert,\wh{H},\ket{\psi(t)})$, contradicting Thesis {\HSF}. 
\end{proof}
\end{replyToObjection}

Now that we understand the role of assigning physical meanings to the operators and the possibilities resulting from this, we can address Objection \ref{obj:unrealistic}. 

\begin{replyToObjection}{obj:unrealistic}
\label{reply:obj:unrealistic}
The translational systems discussed in this article include ideal quantum measurements (Section \sref{s:measurements}), and can be as diverse as the deterministic time-reversible dynamical systems without time loops are (Section \sref{s:dyn_sys}). Moreover, any quantum system becomes a {\TQS} by the mere addition of an ideal clock (Section \sref{s:clock}), or of a massless sterile fermion as in Section \sref{s:weyl}, or by including any translational system, so the entire diversity of possible quantum theories is inherited by such translational systems.
This proves that the Hamiltonian \eqref{eq:time_translation_Hamiltonian} is very versatile, being able to describe the dynamics of very complicated worlds, much like our world, and as diverse as classical or quantum theories can be.
\qed
\end{replyToObjection}

\section{Time operator and clock ambiguity}
\label{s:time-operator-clock}

Sections \sref{s:measurements}--\sref{s:dyn_sys} provide examples of worlds of complexity similar to the complexity of our worlds, and with similar laws, and Theorems \ref{thm:measurement}--\ref{thm:quantum_representation} proved in these Sections show that they are {\TQS}s. 
This does not mean that our world is one of these worlds, but it makes the case that it might be a {\TQS}.
Hopefully these examples may increase the plausibility that our own world is a {\TQS}, but they do not enforce it.

Remarks \ref{rem:tau-and-t} and \ref{rem:interpretation-probability}, allow us to interpret $\wh{\tau}$ as a time operator, for worlds described as {\TQS}s.
However, Theorems \ref{thm:measurement}--\ref{thm:quantum_representation} neither assume, nor do they prove that \emph{we should} interpret $\wh{\tau}$ as the time operator, only that \emph{we can make this additional step}, if it turns out that such a world is a {\TQS}.
Also, the results from this article do not oppose the various other time-related operators proposed in the literature.

If a world is a {\TQS}, it is possible to interpret $\wh{\tau}$ as a time operator.
This provides new avenues to interpret the time-energy uncertainty relation on the same footing as the position-momentum uncertainty, as we expect due to special relativity.
The existence of time operators also allows for the possibility to use the time-extended wavefunction and to take into account the time when a projector associated with a measurement is applied.
But, as we have seen, the physical meaning associated with the other observables is also important. This is necessary in particular for the Page-Wootters formalism:

\begin{objection}[Clock ambiguity problem]
\label{obj:clock_ambiguity}
It was argued in \cite{Albrecht1995TheoryOfEeverythingVsTheoryOfAnything,AlbrechtIglesias2008ClockAmbiguityAndTheEmergenceOfPhysicalLaws,AlbrechtIglesias2012ClockAmbiguityImplicationsNewDevelopments} that the Page-Wootters formalism has a \emph{clock ambiguity} problem.
This ambiguity was found by choosing a different way to express the total Hilbert space $\hilbert_{C+R}=\hilbert_{C}\otimes\hilbert_{R}$ as a tensor product, $\hilbert_{C+R}=\hilbert_{C'}\otimes\hilbert_{R'}$.
\end{objection}

\begin{replyToObjection}{obj:clock_ambiguity}
\label{reply:obj:clock_ambiguity}
This ambiguity is not due to the Page-Wootters formalism, but to the assumption of Thesis {\HSF}.
In addition to the triple $(\hilbert_{C+R},\wh{H}_{C+R},\ket{\psi(t)})$, we need the observables characteristic to the system, as seen in Corollary \ref{thm:HSF}. These include $\wh{\tau}$.
The subspaces of $\hilbert_{C+R}$ of the form $\{\ket{\tau}\}\otimes\hilbert_{R}$, where $\ket{\tau}\in\hilbert_{C}$ is an eigenvector of the operator $\wh{\tau}_{C}$ on $\hilbert_{C}$, are eigenspaces of the operator $\wh{\tau}$ on $\hilbert_{C+R}$.
For a different tensor product decomposition $\hilbert_{C+R}=\hilbert_{C'}\otimes\hilbert_{R'}$, the subspaces of the form $\{\ket{t'}\}\otimes\hilbert_{R'}$, where $\ket{t'}\in\hilbert_{C'}$, are not eigenspaces of the operator $\wh{\tau}$.
In addition, we need all of the observables characterizing the system and their physical meaning, as seen from Corollary \ref{thm:HSF}.
That is, there are observables that apply only to the clock, having the form $\wh{A}_C\otimes\wh{I}_R$, and there are observables that apply only to the rest of the world, of the form $\wh{I}_C\otimes\wh{B}_R$.
They have different physical meanings. The observables of the rest of the world are, in the Page-Wootters formalism, just the observables of our world as described by standard Quantum Mechanics. They determine local algebras of observables on each factor space $\hilbert_C$, respectively $\hilbert_R$. The existence of algebras of local observables enforce a unique factorization, as shown in \cite{ZanardiLidarLloyd2004QuantumTensorProductStructuresAreObservableInduced}.
For the total system $C+R$, the unique factorization into a clock and another system enforced by the local algebras of physical observables is $\hilbert_{C+R}=\hilbert_{C}\otimes\hilbert_{R}$.
Therefore, the ambiguity is not due to the Page-Wootters formalism, as usually believed, but to the assumption that Hilbert Space Fundamentalism is true.
\qed
\end{replyToObjection}

Therefore, the Page-Wootters formalism can be applied unambiguously to the {\TQS} updated with the interpretation of $\wh{\tau}$ as a time operator and the physical interpretations of the relevant observables. It can be applied to Quantum Gravity or to describe experimental situations involving indefinite causal order \cite{BaumannKrummGuerinBrukner2022NoncausalPageWoottersCircuits,SuleymanovCohen2023QuantumFramesOfReferenceAndRelationalFlowOfTime}.

\section{Is there really a problem of unlimited energy extraction?}
\label{s:energy-extraction}

The claim from Objection \ref{obj:negative-energy-descent} is that, if the Hamiltonian of a world is not bounded from below, that world may contain systems that decay toward infinite negative energy states.
Or such a world may contain a system able to charge infinitely many batteries, by extracting energy from another system indefinitely.

\begin{replyToObjection}{obj:negative-energy-descent}
\label{reply:obj:negative-energy-descent}
Objection \ref{obj:negative-energy-descent} may apply to some quantum worlds whose Hamiltonian is not bounded from below, but does it apply to all such worlds?

We consider a world containing an ideal clock $C$, as in Section \sref{s:clock}.
We choose the clock to be a {\TQS} with the Hilbert space $L^2(\R)$, so that it doesn't have interacting parts.
Such a clock is equivalent to a scalar particle whose wavefunction is concentrated at a single point and evolves by translation in the $\tau$-space.
The Hamiltonian from equation \eqref{eq:CR} implies that the clock $C$ and the rest of the world $R$ don't interact.
From Theorem \ref{thm:clock}, the total Hamiltonian has the form $\wh{H}_{C+R} = -i\hbar\frac{\partial\ }{\partial\tau}$, so it is not bounded from below.

Without an interaction between system $R$ and the clock $C$, if the total system $C+R$ would decay indefinitely, this decay should happen within system $R$.
Recall from Reply \ref{reply:obj:clock_ambiguity} from Section \sref{s:time-operator-clock} that the physical meanings of the observables applying either only to the clock (of the form $\wh{A}_C\otimes\wh{I}_R$) or only to the rest of the world (of the form $\wh{I}_C\otimes\wh{B}_R$), enforce a unique factorization $\hilbert_{C+R}=\hilbert_{C}\otimes\hilbert_{R}$ into a clock and the rest of the world.
The reason is that all scenarios of indefinite decay require interactions. Whether a subsystem or a state identified as the vacuum state decays, or a subsystem extracts energy from another subsystem, this involves interactions.
But the clock and the rest of the world $R$ don't interact, and the clock doesn't have interacting subsystems, so all interactions should be confined within system $R$.
But if the energy of system $R$ cannot decay indefinitely, for example because its Hamiltonian is bounded from below, the same should be true for the total system $C+R$, despite the fact that its total Hamiltonian is not bounded from below.
Let us summarize these observations:
\begin{observation}
\label{obs:decay-correspondence}
Whatever way to extract energy one may conceive, the entire work of extraction will happen within system $R$. 
Then, if for whatever reason we think there is a problem with energy extraction in the total system $C+R$ because its Hamiltonian is unbounded, there should be an identical problem for system $R$, which does the actual work.
Therefore, if the Hamiltonian of system $R$ is bounded from below, indefinite energy decay cannot happen in the system $C+R$ either.
\end{observation}

It follows that each example of a Hamiltonian $\wh{H}_R$ that is bounded from below provides a counterexample to Objection \ref{obj:negative-energy-descent}.
In general, whatever reason leads us to believe that Objection  \ref{obj:negative-energy-descent} applies to any {\TQS} in general, the objection is neutralized by choosing a system $R$ not having that problem.

For example, if we think that indefinite energy extraction violates the Second Law of Thermodynamics for the total system $C+R$, system $R$ should be the one violating it, because it is closed in the sense that it does not interact with the clock.
But if $\wh{H}_R$ is bounded from below and system $R$ obeys the Second Law of Thermodynamics, the same must be true of the total system $\wh{H}_{C+R}$. The reason is that entropy is additive, {\ie} the total entropy of two systems is the sum of the entropies of the individual systems \cite{GoldsteinLebowitzTumulkaZanghi2020GibbsAndBoltzmannEntropyInClassicalAndQuantumMechanics}, and the entropy of the clock is constant.

To illustrate Observation \ref{obs:decay-correspondence} in more detail, consider two systems of the world, a system $S$ representing the source of energy that can charge the batteries, and a system $B$ that contains the batteries to be charged. 
We assume that the boundary conditions of the total system $S+B$ are such that the two subsystems start in a state in which the batteries are not charged, but all mechanisms by which they will be charged are in place. 
We assume that the total system evolves in cycles, so that after each cycle, system $S$ returns to a state that is similar to its state at the beginning of the cycle, except that its energy is lower. That is, the state of $S$ is translated along the energy spectrum with a difference of energy $-\Delta E$, where $\Delta E>0$. In the meantime, at the end of each cycle, system $B$ gains an additional amount of energy $+\Delta E$. 

To see this more concretely, let us exhibit in the Hamiltonian of the total system the free and the interaction terms,
\begin{equation}
\label{eq:SB_hamiltonian}
\wh{H}=\wh{H}_S\otimes\wh{I}_B+\wh{I}_S\otimes\wh{H}_B+\wh{H}_{\tn{int}}.
\end{equation}

Let $\ket{E_S,\A_S}_S$ and $\ket{E_B,\A_B}_B$ be the eigenvectors of the free parts of the Hamiltonian, $\wh{H}_S$ and respectively $\wh{H}_B$, where $E_S$ and $E_B$ are eigenvalues of $\wh{H}_S$ and $\wh{H}_B$, and the multiple parameters $\A_S=\(a_{S1},a_{S2},\ldots\)$ and $\A_B=\(a_{B1},a_{B2},\ldots\)$ account for the multiplicity of the energy eigenvalues.
The unit vectors composed of such eigenvectors,
\begin{equation}
\label{eq:SB_basis}
\ket{E_S,\A_S,E_B,\A_B}:=\ket{E_S,\A_S}_S\ket{E_B,\A_B}_B,
\end{equation}
form a basis of $\hilbert$.
All allowed collections of eigenvalues $(E_S,\A_S,E_B,\A_B)$ form an energy configuration space, on which the total state $\ket{\psi(t)}$ can be represented by a wavefunction
\begin{equation}
\label{eq:SB_wavefunction}
\psi_E(E_S,\A_S,E_B,\A_B,t):=\braket{E_S,\A_S,E_B,\A_B}{\psi(t)}.
\end{equation}

To repeatedly extract energy, we assume that the boundary conditions ensure that the wavefunction of the total system becomes, and after $n$ cycles of duration $\Delta t$ each,
\begin{equation}
\label{eq:SB_wavefunction_translated}
\psi_E(E_S-n\Delta E,\A_S,E_B+n\Delta E,\A_B,t_0+n\Delta t)=\psi_E(E_S,\A_S,E_B,\A_B,t_0),
\end{equation}
where $t_0$ is the initial time. In other words, the wavefunction is translated on the energy configuration space with $-n\Delta E$ along the parameter $E_S$ and with $+n\Delta E$ along the parameter $E_B$.
We can assume that system $B$ is infinitely large, and it can use the already charged batteries to produce more batteries to be charged by $S$, so that this process can go on forever.

There seems to be nothing to prevent the possibility in principle of such a scenario.

Now let us particularize this scenario by assuming that $S$ and $R$ are subsystems of a world with a clock, $C+R$.
Then, from Reply \ref{reply:obj:clock_ambiguity} to the clock ambiguity Objection \ref{obj:clock_ambiguity} from Section \sref{s:time-operator-clock},
\begin{enumerate}
	\item 
there is no interaction between $R$ and $C$,
	\item 
and the physical meaning of the observables prevent us from choosing a different factorization than $\hilbert_{C+R}=\hilbert_{C}\otimes\hilbert_{R}$ into a clock and the rest of the world.
\end{enumerate}

Therefore, the systems $S$ and $B$ can be either subsystems of the clock $C$, or subsystems of the rest of the world $R$.

But since for this example we can choose a simple clock with the Hilbert space $L^2(\R)$, the clock doesn't have interacting parts. Therefore, $S$ and $B$ can only be subsystems of $R$.
But if the Hamiltonian $\wh{H}_R$ is bounded from below, $S$ cannot be a subsystem of $R$, because its Hamiltonian is not bounded from below.
Therefore, Objection \ref{obj:negative-energy-descent} doesn't apply to the total system $C+R$ either, despite the fact that its Hamiltonian is not bounded from below.

Systems of the type $C+R$ are the relevant ones if we are interested in the interpretation of the parameter $\tau$ as time in the Page-Wootters formalism.
And since we have seen that Objection \ref{obj:negative-energy-descent} doesn't apply to such systems, this is sufficient to defend the possibility of a time operator against Objection \ref{obj:negative-energy-descent}.
\qed
\end{replyToObjection}

The lesson learned from the example of a system $C+R$ can be applied to other {\TQS}s to exclude indefinite energy decay or extraction:
\begin{observation}
\label{obs:decay-no-ambiguity}
Objection \ref{obj:negative-energy-descent} can be used against a quantum theory (as in Definition \ref{def:quantum-theory}) whose Hamiltonian is not bounded from below only if the mechanism by which the energy decays indefinitely is compatible with the physical observables of the theory.
\end{observation}

For example, the observables of a quantum theory in the Page-Wootters formalism, for example a system $C+R$, are those enforcing $\hilbert_{C+R}=\hilbert_{C}\otimes\hilbert_{R}$ as the only factorization into a clock and the rest of the world.
The physical meaning of these observables prevented the realization of the conditions necessary to apply Objection \ref{obj:negative-energy-descent} to system $C+R$, if these conditions cannot apply as well to system $R$, for example if $\wh{H}_R$ is bounded from below and $R$ satisfies the Laws of Thermodynamics.

In the following we will see a well-known example of a Hamiltonian not bounded from below, which is protected against decay only by the boundary conditions.

\begin{example}
\label{ex:obj:negative-energy-descent:unproblematic}
There are situations in which the conditions of Objection \ref{obj:negative-energy-descent} are satisfied even for Hamiltonians that were not considered problematic in the literature.

Consider a classical particle or a ball in a uniform gravitational field in the direction $z$. The gravitational acceleration along $z$ is $g$, and it is constant. The Hamiltonian function is
\begin{equation}
\label{eq:Hamiltonian-negative-potential-classical}
\mc{H}_S=\frac{p_z^2}{2m}-mgz.
\end{equation}

Because of the potential term, which is negative, this Hamiltonian function can take negative values. In general we can add a constant to make it positive, but in this example there is no such constant for all possible states of the system.
However, we can restrict the state space of the system to exclude negative energy states.

But even if the total energy is positive, if this system evolves freely, it will keep increasing its kinetic energy at the expense of its potential energy, which will decrease indefinitely.

By quantization, we obtain the Hamiltonian operator
\begin{equation}
\label{eq:Hamiltonian-negative-potential-quantum}
\wh{H}_S=-\frac{\hbar^2}{2 m} \frac{\partial^2\ }{\partial z^2}-mg\wh{z}.
\end{equation}

As in the classical case, the energy is transferred more and more from the position degree of freedom $z$ to the momentum one, $p_z$.
It should be possible then to arrange that the degrees of freedom exchanging energy belong to different systems.
For this, we can consider another system $B$, consisting of two fixed vertical rods, and an infinite number of subsystems $\ldots,B_{-1},B_0,B_1,\ldots$, placed equidistantly along one of the rods, as in Figure \ref{fig:energy-extract}. 
This should be possible in principle in Non-Relativistic Quantum Mechanics.

\begin{figure}[!ht]
\begin{center}
\includegraphics[width=0.45\textwidth]{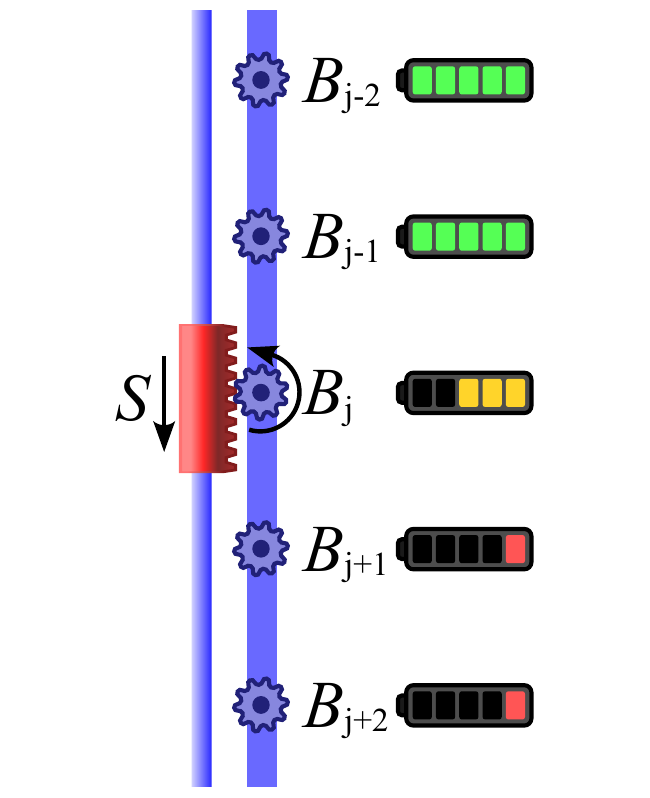}
\caption{A system with non-negative Hamiltonian, able to extract energy indefinitely and to use it to charge an unlimited number of batteries.}
\label{fig:energy-extract}
\end{center}
\end{figure}

System $S$ consists of a rack and a cylinder wrapped around one of the rods, so that it can slide along the rod due to gravity. Each system $B_j$ consists of a cogwheel coupled to a dynamo, so that as $S$ slides down, its rack turns each of the cogwheels. When a cogwheel turns, it turns the dynamo, which charges a battery.
As system $S$ falls, its interaction with the cogwheels prevents it from accelerating, so that its kinetic energy remains bounded.

Therefore, this system satisfies the cyclicity condition \eqref{eq:SB_wavefunction_translated}, and the process can repeat indefinitely.
The total Hamiltonian of the system $S+B$ can be positive, but even so its subsystem $S$ can have negative energies, as long as the total energy is positive. The negative energy manifests only in the potential energy of the subsystem $S$.

Such systems seem to be allowed in both Classical Mechanics and in Non-Relativistic Quantum Mechanics, and they don't seem to contradict the Second Law of Thermodynamics, at least not in an obvious way.
\qed
\end{example}

We can speculate that something forbids systems as in Figure \ref{fig:energy-extract}, but what?
Maybe the statistical interpretation of Thermodynamics implies that the initial conditions of the total system have to be extremely fine-tuned to allow such cyclic processes to be $100\%$ efficient.
For such a cyclic process to exist, maybe the boundary conditions would have to be conspiratorial, so that each cycle succeeds with precision, allowing the indefinite repetition of the extraction.
But maybe it is enough to monitor system $S$ well enough to be able to ensure, at each cycle, that the process does not break down, so that we can extract how much energy we want. This kind of detailed monitoring may, probably, be similar to Maxwell's demon, which was shown, based on Landauer's principle, to increase entropy in order to do the work, and therefore not violate the Second Law \cite{LeffRex2002MaxwellsDemon2}.
However, this may not be a problem if the system is able to extract energy at a sufficient rate to supply the ``Maxwell demon'' with sufficient energy to monitor and correct the process.
Maybe the reason why such systems are not realistic is that they require a sourceless potential uniformly spread everywhere, or some contraptions impossible in practice.
Whatever explanation we try to find seems to lead to the same place: reasonable boundary conditions seem to prevent such a system from extracting energy indefinitely.

In any case, Example \ref{ex:obj:negative-energy-descent:unproblematic} shows that a simple criterion to reject systems based only on the Hamiltonian's spectrum is premature.

We have seen that the fact that the Hamiltonian of a {\TQS} is not bounded from below does not necessarily lead to the problem from Objection \ref{obj:negative-energy-descent}.
We have also seen that there are quantum systems with standard Hamiltonians, which can be used to extract energy indefinitely, unless the boundary conditions somehow prevent this from happening.
Moreover, any system that is not a {\TQS} would conflict with the Second Law, by not being able to ensure irreversibility, as shown in \cite{Stoica2022ProblemOfIrreversibleChangeInQuantumMechanics}.

\section{Conclusions: Empirical adequacy of the time operator and its canonically conjugate translational Hamiltonian}
\label{s:conclusions}

The relativistic spacetime symmetry seems to suggest that, just like the position operators like $\wh{x}$ and the momentum operators like $\wh{p}_x$ are canonically conjugate, so must be the time operator $\wh{\tau}$ and the Hamiltonian $\wh{H}$.
But this implies that $\wh{H}=-i\hbar\frac{\partial\ }{\partial\tau}$.

One may expect that this can be true only for a very simple system, in which time evolution is literally a translation in a configuration space having as coordinate the parameter $\tau$, hence called translational quantum system.
By contrast, the number of fields and parameters of the Standard Model may suggest that our world is governed by a much more complicated Hamiltonian, so this doesn't seem to have the form $\wh{H}=-i\hbar\frac{\partial\ }{\partial\tau}$.
We have seen that this judgment would be too rushed, because whatever complexity we expect to find in our world, it can be found in worlds being translational quantum systems.
Any quantum world, no matter its Hamiltonian, when extended with a clock turns into a world that can be described as a translational quantum system (Section \sref{s:clock}).
The same is true for any quantum world containing a sterile massless fermion (Section \sref{s:weyl}), or a system measuring another system according to the standard measurement scheme (Section \sref{s:measurements}), and to any quantum world that happens to be the representation of a classical system evolving without cycles (Section \sref{s:dyn_sys}).

We have seen that the resolution of the apparent contradiction between the simplicity of the translational Hamiltonian and the sophistication of the world is that the sophistication doesn't come from the Hamiltonian itself, but from both the Hamiltonian and the physical meaning of the other observables (Section \sref{s:HSF}).

Another point raised against translational quantum systems is that they allow, supposedly, the extraction of indefinite amounts of energy.
The examples from this article show that if this were true, so must be the case for any quantum or classical system (Section \sref{s:energy-extraction}).
Therefore, these objections against time as a quantum operator are dissolved (Sections \sref{s:energy-extraction}, \sref{s:HSF}, and \sref{s:time-operator-clock}).

Moreover, the rich examples of translational quantum systems from this article provide infinitely many counterexamples to the Hilbert Space Fundamentalism Thesis (\HSF) (Section \sref{s:HSF}).
The resolution, as seen, is to take into account the physical meaning of the observables, because this is the source of diversity and sophistication of these examples.
Moreover, the same solution resolves the clock ambiguity objection against the Page-Wootters formalism (Section \sref{s:time-operator-clock}).

The Hamiltonian $\wh{H}=-i\hbar\frac{\partial\ }{\partial\tau}$ is sufficiently versatile to describe the dynamics of worlds of unlimited complexity and extreme diversity. Whether it applies to our world remains to be seen.
In case it does, this would not make it a ``theory of everything'', because, as seen in Corollary \ref{thm:HSF}, the physical meaning of the observables, and the dependency of the Hamiltonian of these observables, are still needed, and these observables represent most of the physical content of the theory.

Therefore, there is sufficient evidence to at least take seriously the possibility of a world governed by a deceptively simple Hamiltonian having the time operator as its conjugate.

\section*{Declarations}

\ \par
\textbf{Funding:} This research received no external funding.

\textbf{Data Availability Statement:} Data sharing is not applicable.

\textbf{Conflicts of Interest:} The author declares no conflicts of interest.

\textbf{Experiments:} The author declares that no experiments on humans or animals were conducted.

\textbf{Acknowledgement} The author thanks Basil Altaie, Almut Beige, Eliahu Cohen, Ismael Paiva, Ashmeet Singh, Michael Suleymanou, and anonymous referees, for their valuable comments and suggestions offered to a previous version of the manuscript. Nevertheless, the author bears full responsibility for the article.


\section*{References}



\begin{thebibliography}{10}

\bibitem{Pauli1980GeneralPrinciplesOfQuantumMechanics}
W.~Pauli.
\newblock {\em General principles of quantum mechanics}.
\newblock Springer, Berlin Heidelberg New York, 1990.

\bibitem{Heisenberg1927Uncertainty}
W.~Heisenberg.
\newblock {\"Uber den anschaulichen Inhalt der quantentheoretischen Kinematik
  und Mechanik}.
\newblock {\em Zeitschrift f{\"{u}}r Physik}, 43(3--4):172--198, 1927.

\bibitem{MandelstamTamm1945TheUncertaintyRelationBetweenEnergyAndTimeInNonrelativisticQuantumMechanics}
L.~Mandelstam and Ig. Tamm.
\newblock The uncertainty relation between energy and time in nonrelativistic
  quantum mechanics.
\newblock {\em J. Phys. (USSR)}, 9:249--254, 1945.

\bibitem{Razavy1971TimeOfArrivalOperator}
M.~Razavy.
\newblock Time of arrival operator.
\newblock {\em Canad. J. Phys.}, 49(24):3075--3081, 1971.

\bibitem{MugaEtal2007TimeInQuantumMechanicsI}
G.~Muga, R.S. Mayato, and I.~Egusquiza.
\newblock Time in quantum mechanics, 2007.

\bibitem{Hartman1962TunnelingOfAWavePacket}
T.E. Hartman.
\newblock Tunneling of a wave packet.
\newblock {\em J. Appl. Phys.}, 33(12):3427--3433, 1962.

\bibitem{Wigner1988OnTheTimeEnergyUncertaintyRelation}
E.P. Wigner.
\newblock On the time--energy uncertainty relation.
\newblock In {\em Special relativity and quantum theory: {A} collection of
  papers on the {P}oincar{\'e} group\/}
  \cite{NozKim1988SpecialRelativityAndQuantumTheoryACollectionOfPapersOnThePoincareGroup},
  pages 199--209.

\bibitem{AharonovBohm1961TimeInTheQuantumTheoryAndTheUncertaintyRelationForTimeAndEnergy}
Y.~Aharonov and D.~Bohm.
\newblock Time in the quantum theory and the uncertainty relation for time and
  energy.
\newblock {\em Phys. Rev.}, 122(5):1649, 1961.

\bibitem{Fock1962CriticismOfAnAttemptToDisproveTheUncertaintyRelationBetweenTimeAndEnergy}
V.A. Fock.
\newblock Criticism of an attempt to disprove the uncertainty relation between
  time and energy.
\newblock {\em Sov. Phys. JETP}, 15(4):784--786, 1962.

\bibitem{BuschGrabowskiLahti1995OperationalQuantumPhysics}
P.~Busch, M.~Grabowski, and P.~Lahti.
\newblock {\em Operational quantum physics}.
\newblock Springer, Berlin, Heidelberg, 1995.

\bibitem{Busch2002Time-energy-uncertainty-relation}
P.~Busch.
\newblock The time-energy uncertainty relation.
\newblock In {\em Time in quantum mechanics\/}
  \cite{MugaEtal2007TimeInQuantumMechanicsI}, pages 73--105.

\bibitem{Hilgevoord1996TheUncertaintyPrincipleForEnergyAndTimeI}
J.~Hilgevoord.
\newblock The uncertainty principle for energy and time.
\newblock {\em Am. J. Phys.}, 64(12):1451--1456, 1996.

\bibitem{Hilgevoord1998TheUncertaintyPrincipleForEnergyAndTimeII}
J.~Hilgevoord.
\newblock The uncertainty principle for energy and time. {II}.
\newblock {\em Am. J. Phys.}, 66(5):396--402, 1998.

\bibitem{Hilgevoord2005TimeInQuantumMechanicsAStoryOfConfusion}
J.~Hilgevoord.
\newblock Time in quantum mechanics: a story of confusion.
\newblock {\em Stud. Hist. Philos. Mod. Phys.}, 36(1):29--60, 2005.

\bibitem{Jammer1974ThePhilosophyOfQuantumMechanics}
M.~Jammer.
\newblock {\em The Philosophy of Quantum Mechanics}.
\newblock John Wiley \& Sons, Inc., New York, 1974.

\bibitem{BauerMello1978TheTimeEnergyUncertaintyRelation}
M.~Bauer and P.A. Mello.
\newblock The time-energy uncertainty relation.
\newblock {\em Annals of Physics}, 111(1):38--60, 1978.

\bibitem{Busch1990OnTheEnergyTimeUncertaintyRelationI}
P.~Busch.
\newblock On the energy-time uncertainty relation. {P}art {I}: {D}ynamical time
  and time indeterminacy.
\newblock {\em Found. Phys.}, 20(1):1--32, 1990.

\bibitem{Busch1990OnTheEnergyTimeUncertaintyRelationII}
P.~Busch.
\newblock On the energy-time uncertainty relation. {P}art {II}: {P}ragmatic
  time versus energy indeterminacy.
\newblock {\em Found. Phys.}, 20(1):33--43, 1990.

\bibitem{AltaieBeigeHodgson2022TimeAndQuantumClocksAReviewOfRecentDevelopments}
M.B. Altaie, A.~Beige, and D.~Hodgson.
\newblock Time and quantum clocks: a review of recent developments.
\newblock {\em Arxiv
  \href{http://arxiv.org/abs/2203.12564}{arXiv:quant-ph/2203.12564}}, 2022.

\bibitem{NozKim1988SpecialRelativityAndQuantumTheoryACollectionOfPapersOnThePoincareGroup}
M.~Noz and Y.S. Kim.
\newblock Special relativity and quantum theory: {A} collection of papers on
  the {P}oincar{\'e} group, 1988.

\bibitem{MugaEtal2009TimeInQuantumMechanicsII}
G.~Muga, A.~Ruschhaupt, and A.~del Campo.
\newblock Time in quantum mechanics, 2009.

\bibitem{Hall2013QuantumTheoryForMathematicians}
B.C. Hall.
\newblock {\em Quantum theory for mathematicians}, volume 267.
\newblock Springer, 2013.

\bibitem{Dewitt1967QuantumTheoryOfGravityI_TheCanonicalTheory}
B.S. DeWitt.
\newblock Quantum theory of gravity. {I}. {T}he canonical theory.
\newblock {\em Phys. Rev.}, 160(5):1113, 1967.

\bibitem{PageWootters1983EvolutionWithoutEvolution}
D.N. Page and W.K. Wootters.
\newblock Evolution without evolution: {D}ynamics described by stationary
  observables.
\newblock {\em Phys. Rev. D}, 27(12):2885, 1983.

\bibitem{Kuchar2011TimeAndInterpretationsOfQuantumGravity}
K.V. Kucha{\v{r}}.
\newblock Time and interpretations of quantum gravity.
\newblock {\em Int. J. Mod. Phys. D.}, 20(supp01):3--86, 2011.

\bibitem{Isham1993CanonicalQuantumGravityAndTheProblemOfTime}
C.J. Isham.
\newblock Canonical quantum gravity and the problem of time.
\newblock In {\em Integrable systems, quantum groups, and quantum field
  theories}, pages 157--287. Kluwer Academic Publishers, 1993.

\bibitem{Hardy2009QuantumGravityComputersOnTheTheoryOfComputationWithIndefiniteCausalStructure}
L.~Hardy.
\newblock Quantum gravity computers: {O}n the theory of computation with
  indefinite causal structure.
\newblock In {\em Quantum reality, relativistic causality, and closing the
  epistemic circle. {E}ssays in honour of {A}bner {S}himony}, volume~73 of {\em
  Western Ontario Series in Philosophy of Science}, pages 379--401. Springer,
  2009.

\bibitem{OreshkovCostaBrukner2012QuantumCorrelationsWithNoCausalOrder}
O.~Oreshkov, F.~Costa, and {\v{C}}~Brukner.
\newblock Quantum correlations with no causal order.
\newblock {\em Nat. Comm.}, 3(1):1092, 2012.

\bibitem{ProcopioEtAl2015ExperimentalSuperpositionOfOrdersOfQuantumGates}
L.Procopio, A.~Moqanaki, M.~Ara{\'u}jo, F.~Costa, I.A. Calafell, E.~Dowd,
  D.~Hamel, L.~Rozema, {\v{C}}~Brukner, and P.~Walther.
\newblock Experimental superposition of orders of quantum gates.
\newblock {\em Nature communications}, 6(1):7913, 2015.

\bibitem{BaumannKrummGuerinBrukner2022NoncausalPageWoottersCircuits}
V.~Baumann, M.~Krumm, P.A. Gu{\'e}rin, and {\v{C}}~Brukner.
\newblock Noncausal {P}age--{W}ootters circuits.
\newblock {\em Phys. Rev. Res.}, 4(1):013180, 2022.

\bibitem{SuleymanovCohen2023QuantumFramesOfReferenceAndRelationalFlowOfTime}
M.~Suleymanov and E.~Cohen.
\newblock Quantum frames of reference and the relational flow of time.
\newblock {\em Eur. Phys. J. Special Topics}, pages 1--13, 2023.

\bibitem{Stoica2022ProblemOfIrreversibleChangeInQuantumMechanics}
O.C. Stoica.
\newblock The problem of irreversible change in quantum mechanics.
\newblock 2022.

\bibitem{Eddington1928NatureOfThePhysicalWorld}
A.~Eddington.
\newblock {\em The nature of the physical world}.
\newblock Dent, London, 1928.

\bibitem{Malament1996InDefenseOfDogmaWhyCannotBeARelativisticQuantumMechanicsOfParticles}
D.B. Malament.
\newblock In defense of dogma: {W}hy there cannot be a relativistic quantum
  mechanics of (localizable) particles.
\newblock In {\em Perspectives on Quantum Reality: {N}on-Relativistic,
  Relativistic, and Field-Theoretic}, pages 1--10. Springer, 1996.

\bibitem{AulettaFortunatoParisi2009QuantumMechanics}
G.~Auletta, Fortunato M, and Parisi G.
\newblock {\em Quantum mechanics}.
\newblock Cambridge University Press, Cambridge, UK, 2009.

\bibitem{SrinivasVijayalakshmi1981TimeOfOccurrenceInQuantumMechanics}
M.D. Srinivas and R.~Vijayalakshmi.
\newblock The 'time of occurrence' in quantum mechanics.
\newblock {\em Pramana}, 16:173--199, 1981.

\bibitem{Mittelstaedt2004InterpretationOfQMAndMeasurementProcess}
P.~Mittelstaedt.
\newblock {\em The interpretation of quantum mechanics and the measurement
  process}.
\newblock Cambridge University Press, Cambridge, UK, 2004.

\bibitem{Stone1932OnOneParameterUnitaryGroupsInHilbertSpace}
M.H. Stone.
\newblock On one-parameter unitary groups in {H}ilbert space.
\newblock {\em Ann. Math.}, pages 643--648, 1932.

\bibitem{BrinStuck2002IntroductionToDynamicalSystems_ABSTRACT}
M.~Brin and G.~Stuck.
\newblock {\em Introduction to dynamical systems}.
\newblock Cambridge University Press, 2002.

\bibitem{Koopman1931HamiltonianSystemsAndTransformationInHilbertSpace}
B.O. Koopman.
\newblock Hamiltonian systems and transformation in {H}ilbert space.
\newblock {\em Proc. Nat. Acad. Sci. U.S.A.}, 17(5):315, 1931.

\bibitem{vonNeumann1932KoopmanMethod}
J.~v.~Neumann.
\newblock Zur {O}peratorenmethode in der klassischen {M}echanik.
\newblock {\em Ann. Math.}, pages 587--642, 1932.

\bibitem{Schonberg1952TheActualKoopmanRepApplicationOf2ndQuantizationToClassicalI}
M~Sch{\"o}nberg.
\newblock Application of second quantization methods to the classical
  statistical mechanics.
\newblock {\em Il Nuovo Cimento (1943-1954)}, 9:1139--1182, 1952.

\bibitem{Schonberg1953TheActualKoopmanRepApplicationOf2ndQuantizationToClassicalII}
M~Sch{\"o}nberg.
\newblock Application of second quantization methods to the classical
  statistical mechanics {(II)}.
\newblock {\em Il Nuovo Cimento (1943-1954)}, 10:419--472, 1953.

\bibitem{Carroll2021RealityAsAVectorInHilbertSpace}
S.M. Carroll.
\newblock Reality as a vector in {H}ilbert space.
\newblock Technical Report CALT-TH-2021-010, Cal-Tech, 2021.

\bibitem{Piazza2010GlimmersOfAPreGeometricPerspective}
F.~Piazza.
\newblock Glimmers of a pre-geometric perspective.
\newblock {\em Found. Phys.}, 40:239--266, 2010.

\bibitem{Tegmark2015ConsciousnessAsAStateOfMatter}
M.~Tegmark.
\newblock Consciousness as a state of matter.
\newblock {\em Chaos Solitons Fractals}, 76:238--270, 2015.

\bibitem{CarrollSingh2019MadDogEverettianism}
S.M. Carroll and A.~Singh.
\newblock Mad-dog {E}verettianism: {Q}uantum {M}echanics at its most minimal.
\newblock In A.~Aguirre, B.~Foster, and Z.~Merali, editors, {\em What is
  Fundamental?}, pages 95--104. Springer, 2019.

\bibitem{Aaronson2021TheZenAntiInterpretationOfQuantumMechanics}
S.~Aaronson.
\newblock The {Z}en anti-interpretation of {Q}uantum {M}echanics.
\newblock {\em
  \href{https://www.scottaaronson.com/blog/?p=5359}{www.scottaaronson.com/blog/?p=5359}},
  2021.

\bibitem{ZanardiLidarLloyd2004QuantumTensorProductStructuresAreObservableInduced}
P.~Zanardi, D.A. Lidar, and S.~Lloyd.
\newblock Quantum tensor product structures are observable induced.
\newblock {\em Phys. Rev. Lett.}, 92(6):060402, 2004.

\bibitem{Zurek1998DecoherenceEinselectionAndTheExistentialInterpretation}
W.H. Zurek.
\newblock {Decoherence, Einselection, and the Existential Interpretation (The
  Rough Guide)}.
\newblock {\em Phil.\ Trans.\ Roy.\ Soc.\ London Ser. A}, A356:1793, 1998.
\newblock \href{http://arxiv.org/abs/quant-ph/9805065}{arXiv:quant-ph/9805065}.

\bibitem{Vaidman2016AllIsPsi}
L.~Vaidman.
\newblock All is $\psi$.
\newblock {\em J. Phys. Conf. Ser.}, 701:012020, 2016.

\bibitem{SEP-Vaidman2021MWI}
L.~Vaidman.
\newblock Many-worlds interpretation of quantum mechanics.
\newblock In E.N. Zalta, editor, {\em The {Stanford} Encyclopedia of
  Philosophy}. Metaphysics Research Lab, Stanford University, 2021.
\newblock https://plato.stanford.edu/entries/qm-manyworlds/, last accessed
  \today.

\bibitem{Albrecht1995TheoryOfEeverythingVsTheoryOfAnything}
A.~Albrecht.
\newblock The theory of everything vs the theory of anything.
\newblock In F.~Occhionero, editor, {\em Birth of the Universe and Fundamental
  Physics}, volume 455, pages 321--332. Springer, 1995.

\bibitem{AlbrechtIglesias2008ClockAmbiguityAndTheEmergenceOfPhysicalLaws}
A.~Albrecht and A.~Iglesias.
\newblock Clock ambiguity and the emergence of physical laws.
\newblock {\em Phys. Rev. D}, 77(6):063506, 2008.

\bibitem{AlbrechtIglesias2012ClockAmbiguityImplicationsNewDevelopments}
A.~Albrecht and A.~Iglesias.
\newblock The clock ambiguity: {I}mplications and new developments.
\newblock In {\em The Arrows of Time}, volume 172, pages 53--68. Springer,
  2012.

\bibitem{MarlettoVedral2017EvolutionWithoutEvolutionAndWithoutAmbiguities}
C.~Marletto and V.~Vedral.
\newblock Evolution without evolution and without ambiguities.
\newblock {\em Phy. Rev. D}, 95(4):043510, 2017.

\bibitem{CotlerEtAl2019LocalityFromSpectrum}
J.S. Cotler, G.R. Penington, and D.H. Ranard.
\newblock Locality from the spectrum.
\newblock {\em Comm. Math. Phys.}, 368(3):1267--1296, 2019.

\bibitem{Stoica2022SpaceThePreferredBasisCannotUniquelyEmergeFromTheQuantumStructure}
O.C. Stoica.
\newblock 3d-space and the preferred basis cannot uniquely emerge from the
  quantum structure.
\newblock {\em Adv. Theor. Math. Phys.}, 26(10):3895---3962, 2022.
\newblock \href{http://arxiv.org/abs/2102.08620}{arXiv:2102.08620}.

\bibitem{Stoica2024DoesTheHamiltonianDetermineTheTPSAndThe3dSpace}
O.C. Stoica.
\newblock Does the {H}amiltonian determine the tensor product structure and the
  3d space?
\newblock 2024.

\bibitem{Giddings2019QuantumFirstGravity}
S.B. Giddings.
\newblock Quantum-first gravity.
\newblock {\em Found. Phys.}, 49(3):177--190, 2019.

\bibitem{Schmelzer2009WhyTheHamiltonOperatorAloneIsNotEnough}
Ilja Schmelzer.
\newblock Why the {H}amilton operator alone is not enough.
\newblock {\em Found. Phys.}, 39(5):486--498, 2009.

\bibitem{Maldacena1999AdSCFTInitial}
J~Maldacena.
\newblock The large-{N} limit of superconformal field theories and
  supergravity.
\newblock {\em Int. J. Theor. Phys}, 38(4):1113--1133, 1999.

\bibitem{Kaplan2016LecturesOnAdSCFTFromTheBottomUp}
J.~Kaplan.
\newblock Lectures on {AdS/CFT} from the bottom up.
\newblock 2016.
\newblock
  \href{https://sites.krieger.jhu.edu/jared-kaplan/files/2016/05/AdSCFTCourseNotesCurrentPublic.pdf}{https://sites.krieger.jhu.edu/jared-kaplan/files/2016/05/AdSCFTCourseNotesCurrentPublic.pdf},
  last accessed \today.

\bibitem{AlmheiriDongHarlow2015BulkLocalityAndQuantumErrorCorrectionInAdSCFT}
A.~Almheiri, Xi~Dong, and D.~Harlow.
\newblock Bulk locality and quantum error correction in ads/cft.
\newblock {\em JHEP}, 2015(4):1--34, 2015.

\bibitem{Stoica2023PrinceAndPauperQuantumParadoxHilbertSpaceFundamentalism}
O.C. Stoica.
\newblock The prince and the pauper. {A} quantum paradox of {H}ilbert-space
  fundamentalism.
\newblock 2023.

\bibitem{GoldsteinLebowitzTumulkaZanghi2020GibbsAndBoltzmannEntropyInClassicalAndQuantumMechanics}
S.~Goldstein, J.L. Lebowitz, R.~Tumulka, and N.~Zangh{\`\i}.
\newblock Gibbs and {B}oltzmann entropy in classical and quantum mechanics.
\newblock pages 519--581. World Scientific, 2020.

\bibitem{LeffRex2002MaxwellsDemon2}
H.S. Leff and A.F. Rex.
\newblock {\em Maxwell's Demon 2. {E}ntropy, Classical and Quantum Information,
  Computing}.
\newblock CRC Press, 2002.

\end{thebibliography}
\end{document}